\documentclass[sigconf]{acmart}

\AtBeginDocument{%
  \providecommand\BibTeX{{%
    \normalfont B\kern-0.5em{\scshape i\kern-0.25em b}\kern-0.8em\TeX}}}

\copyrightyear{2022}
\acmYear{2022}
\setcopyright{acmlicensed}\acmConference[WWW '22]{Proceedings of the ACM
	Web Conference 2022}{April 25--29, 2022}{Virtual Event, Lyon, France}
\acmBooktitle{Proceedings of the ACM Web Conference 2022 (WWW '22), April
	25--29, 2022, Virtual Event, Lyon, France}
\acmPrice{15.00}
\acmDOI{10.1145/3485447.3512220}
\acmISBN{978-1-4503-9096-5/22/04}

\acmSubmissionID{w10fp0845}

\usepackage{grffile}
\usepackage{xcolor}
\usepackage{multirow}
\usepackage{array}
\usepackage{amsmath}
\usepackage{amsthm}
\usepackage{graphicx}
\usepackage{caption}
\usepackage{subcaption}
\usepackage[ruled, vlined, linesnumbered]{algorithm2e}
\SetKwInput{KwInput}{Input}  
\newtheorem{definition}{Definition}
\newtheorem{theorem}{Theorem}[section]
\newtheorem{corollary}{Corollary}[theorem]
\newtheorem{lemma}[theorem]{Lemma}
\newcommand{\norm}[1]{\left\lVert#1\right\rVert}
\newcommand{\real}{\mathbb{R}}
\newcommand{\para}[1]{\noindentparagraph{\textbf{\textup{#1}}}}
\newcommand{\eps}{\varepsilon}
\newcommand{\var}{\mathsf{Var}}
\newcommand{\E}{\mathbb{E}}
\begin{document}

\title{Compressive Sensing Approaches for Sparse Distribution Estimation Under Local Privacy}

\author{Zhongzheng Xiong}
\email{zzxiong21@m.fudan.edu.cn}
\affiliation{%
	\institution{School of Data Science,}
	\institution{Fudan University}
	\city{Shanghai}
	\country{China}
}

\author{Jialin Sun}
\email{sunjl20@fudan.edu.cn}
\affiliation{%
	\institution{School of Data Science,}
	\institution{Fudan University}
	\city{Shanghai}
	\country{China}
}

\author{Xiaojun Mao}
\email{maoxj@sjtu.edu.cn}
\affiliation{
	\institution{School of Mathematical Sciences,}
	\institution{Shanghai Jiao Tong University}
	\city{Shanghai}
	\country{China}
}

\author{Jian Wang}
\email{jian\_wang@fudan.edu.cn}
\affiliation{%
	\institution{School of Data Science,}
	\institution{Fudan University}
	\city{Shanghai}
	\country{China}
}

\author{Ying Shan}
\email{yingsshan@tencent.com}
\affiliation{
	\institution{Tencent}
	\city{Shenzhen}
	\country{China}
}

\author{Zengfeng Huang}
\email{huangzf@fudan.edu.cn}
\authornote{Corresponding author}
\affiliation{%
	\institution{School of Data Science,}
	\institution{Fudan University}
	\city{Shanghai}
	\country{China}
}


\begin{abstract}
 {Recent years, local differential privacy (LDP) has been adopted by many web service providers like Google \cite{erlingsson2014rappor}, Apple \cite{apple2017privacy} and Microsoft \cite{bolin2017telemetry} to collect and analyse users' data privately}. In this paper, we consider the problem of discrete distribution estimation under local differential privacy constraints. Distribution estimation is one of the most fundamental estimation problems, which is widely studied in both non-private and private settings. In the local model, private mechanisms with provably optimal sample complexity are known. However, they are optimal only in the worst-case sense; their sample complexity is proportional to the size of the entire universe, which could be huge in practice. In this paper, we consider sparse or approximately sparse (e.g.\ highly skewed) distribution, and show that the number of samples needed could be significantly reduced. This problem has been studied recently \cite{acharya2021estimating}, but they only consider strict sparse distributions and the high privacy regime. We propose new privatization mechanisms based on compressive sensing. Our methods work for approximately sparse distributions and medium privacy, and have optimal sample and communication complexity.
\end{abstract} 

\ccsdesc[500]{Security and privacy~Privacy-preserving protocols}
\ccsdesc[500]{Mathematics of computing~Density estimation}

\keywords{locally differential privacy, sparse distribution estimation, compressive sensing.}


\maketitle

\section{Introduction}
Discrete distribution estimation \cite{kamath2015learning, lehmann2006theory,kairouz2016discrete} from samples is a fundamental problem in statistical analysis. In the traditional statistical setting, the primary goal is to achieve best trade-off between sample complexity and estimation accuracy. In many modern data analytical applications, the raw data often contains sensitive information, e.g.\ medical data of patients, and it is prohibitive to release them without appropriate privatization. Differential privacy is one of the most popular and powerful definitions of privacy \cite{dwork2006calibrating}. Traditional centralized model assumes there is a trusted data collector. In this paper, we consider \textit{locally differential privacy} (LDP) \cite{warner1965randomized, kasiviswanathan2011can, beimel2008distributed}, where users privatize their data before releasing it so as to keep their personal data private even from data collectors. {Recently, LDP has been deployed in real world online platforms by several technology organizations including Google \cite{erlingsson2014rappor}, Apple \cite{apple2017privacy} and Microsoft \cite{bolin2017telemetry}. For example, Google deployed their RAPPOR system \cite{erlingsson2014rappor} in Chrome browser for analyzing web browsing behaviors of users in a privacy-preserving manner. LDP has become the standard privacy model for large-scale distributed applications and LDP algorithms are now being used by hundreds of millions of users daily.} 

We study the discrete distribution estimation problem under LDP constraints. The main theme in private distribution estimation is to optimize statistical and computational efficiency under privacy constraints. Given a privacy parameter, the goal to achieve best tradeoff between estimation error and sample complexity. In the local model, the communication cost and computation time are also important complexity parameters.  This problem has been widely studied in the local model recently \cite{warner1965randomized, duchi2013local, erlingsson2014rappor, kairouz2014extremal, wang2016mutual, pastore2016locally, kairouz2016discrete, acharya2018hadamard,ye2018optimal, bassily2019linear, acharya2018hadamard}.  Thus far, the worst-case sample complexity, i.e., the minimum number of samples needed to achieve a desired accuracy for the worst-case distribution, has been well-understood \cite{wang2016mutual, ye2018optimal, acharya2018hadamard}. However, worst-case behaviors are often not indicative of their performance in practice; real-world inputs often contain special structures that allow one to bypass such worst-case barriers.

In this paper, we consider sparse or approximately sparse distributions $p\in \real^k$, which are perhaps the most natural structured distributions. Let $k$ be the ambient dimensionality of the distribution. The goal in this setting is to achieve \emph{sublinear} (in $k$) sample complexity.
This problem has been studied in \cite{acharya2021estimating} very recently. Their method first applies one-bit Hadamard response from~\cite{acharya2018hadamard}, and then projects the final estimate to the set of sparse distributions; it is proved that this simple idea leads to sample complexity that only depends on the sparsity $s$. However, there are still several problems left unresolved. First, the theoretical results in \cite{acharya2021estimating} only hold for strictly sparse distributions, which is too restrictive for many applications. Second, they only consider the high privacy regime. A more subtle issue is that, from their algorithm and analyses, the number of samples needed is implicitly assumed to be larger than $k$. This is because one-bit HR needs to partition the samples into more than $k$ groups of the same size; otherwise the estimation procedure is not well-defined. Therefore, their technique cannot achieve \emph{sublinear} sample complexity even if the distribution is extremely sparse.  It is unclear to us whether their projection-based techniques can be modified to resolve all these problems.
In this paper, we take a different approach, which resolves the above issues in a unified way. Our contributions are summarized as follows.
\begin{enumerate}
	\item  We propose novel privatization schemes based on \emph{compressive sensing} (CS). Our new algorithms have optimal sample and {communication complexity} simultaneously for sparse distribution estimation. As far as we know, these are the first LDP schemes that achieve this; and this is the first work to apply CS techniques in LDP distribution estimation.
	\item Applying standard results in CS theory, our method is immediately applicable to estimating approximately sparse distributions.
	\item We also generalize our techniques to handle medium privacy regimes using ideas from model-based compressive sensing.
\end{enumerate}
Our main idea is to do privatization and dimensionality reduction simultaneously, and then perform distribution estimation in the lower dimensional space. This can reduce the sample complexity because the estimation error depends on the dimensionality of the distribution. The original distribution is then recovered from the low-dimensional one using tools from compressive sensing. We call this technique \emph{compressive privatization} (CP). 
\subsection{Problem Definition and Results}
We consider $k$-ary discrete distribution estimation. W.l.o.g., we assume the target distribution is defined on the universe $\mathcal{X} =[k]:= [1, 2,\cdots, k]$, which can be viewed as a $k$-dimensional vector $p\in \real^k$ with $\norm{p}_1=1$. Let $\Delta_k$ be the set of all $k$-ary discrete distributions. Given $n$ i.i.d.\ samples, $X_1, \cdots, X_n$, drawn from the unknown distribution $p$, the goal is to provide an estimator $\hat{p}$ such that $d(p,\hat{p})$ is minimized, where $d(,)$ is typically the $\ell_1$ or $\ell_2$ norm.

\para{Local privacy.} In the local model, each $X_i$ is held by a different user. Each user will only send a privatized version of their data to the central server, who will then produces the final estimate. A privatization mechanism is a randomized mapping $Q$ that maps $x \in \mathcal{X}$ to $y \in \mathcal{Y}$ with probability $Q(y|x)$ for some output set $\mathcal{Y}$. The mapping $Q$ is said to be $\varepsilon$-locally differential private (LDP) \cite{duchi2013local} if for all $x, x' \in \mathcal{X}$ and $y \in \mathcal{Y}$, we have
	$$\frac{Q(y|x)}{Q(y|x')} \leq e^\varepsilon.$$

\para{LDP distribution estimation.}
Let $Y=(Y_1, Y_2, \cdots Y_n)\in \mathcal{Y}^n$ be the privatized samples obtained by applying $Q$ on $X=(X_1, \cdots, X_n)$. Given privacy parameter $\eps$, the goal of LDP distribution estimation is to design an $\eps$-LDP mapping $Q$ and a corresponding estimator $\hat{p}: \mathcal{Y}^n \rightarrow \Delta_k$, such that $\E[ d(p,\hat{p})]$ is minimized. Given $\eps$ and $\alpha$, we are most interested in the number of samples needed (as a function of $\eps$ and $\alpha$) to assure $\eps$-LDP and $\E[ d(p,\hat{p})]\le \alpha$. 

\para{Sparsity.} A discrete distribution $p\in \triangle_k$ is called $s$-sparse if the number of non-zeros in $p$ is at most $s$. Let $[p]_s$ be the $s$-sparse vector that contains the top-$s$ entries of $p$. We say $p$ is approximately $(s,\lambda)$-sparse, if $\|p-[p]_s\|_1\le \lambda$.

\para{Our results.} For the high privacy regime, i.e., $\eps =O(1)$, existing studies \cite{warner1965randomized, kairouz2016discrete, acharya2018hadamard, erlingsson2014rappor, ye2018optimal, wang2016mutual} have achieved optimal sample complexity, which is $\Theta\left(\frac{k^2}{\alpha^2 \varepsilon^2}\right)$ for $\ell_1$ norm error and $\Theta\left(\frac{k}{\alpha^2 \epsilon^2}\right)$ for $\ell_2$ norm error. These worst-case optimal bounds have a dependence on $k$. 
Our result (informal) for the high privacy regime is summarized as follows; see Theorem~\ref{thm:RandomMain_onebit} for exact bounds and results for approximately sparse distributions. 
\begin{theorem}[Informal]
	For any $0< \eps < 1$ and $\alpha>0$, there is an $\eps$-LDP scheme $Q$, which produces an estimator $\hat{p}$ with error guarantee $d(p,\hat{p})\le \alpha$. If $p$ is $s$ sparse, then the sample complexity of $Q$ is $O\left(\frac{s^2\log (k/s)}{\eps^2\alpha^2}\right)$ for $\ell_1$ error and $O\left(\frac{s\log (k/s)}{\eps^2 \alpha^2}\right)$ for $\ell_2$ error.
\end{theorem}
We provide two different sample optimal privatization methods; the first one has one-bit communication cost and the other one is symmetric (i.e. all users perform the same privatization scheme) but at the cost of using logarithmic communication. Symmetric mechanisms could be beneficial in some distributed settings; it is proved that the communication cost overhead cannot be avoided \cite{acharya2019communication}.
Our CS-based technique can be extended to the medium privacy regime with $1 \le \eps \le \log s$ (see Section \ref{sec:median}). The result is summarized in the following theorem. See Theorem \ref{thm: rcp_medium_result} for exact bounds.
\begin{theorem}[Informal]
	For any $1 \le e^\eps, 2^b \le s$ and $\alpha > 0$, there is an $\eps$-LDP scheme $Q$, which produces an estimator $\hat{p}$ with error guarantee $d(p , \hat{p}) \leq \alpha$ with communication no more than $b$ bits. If $p$ is $s$ sparse, then the sample complexity of $Q$ is $O(\frac{s^2\log {k}/{s}}{\min\{e^\eps, 2^b\}\alpha^2})$ for $\ell_1$ error and $O(\frac{s\log {k}/{s}}{\min\{e^\eps, 2^b\}\alpha^2})$ for $\ell_2$ error. 
\end{theorem}
This result provides a characterization on the relationship between accuracy, privacy, and communication, which is nearly tight.  A tight and complete characterization for dense distribution estimation was obtained in \cite{chen2020breaking}. 
\subsection{Related work}
Differential privacy is the most widely adopted notion of privacy \cite{dwork2006calibrating}; a large body of literature exists (see e.g.\ \cite{dwork2014algorithmic} for a comprehensive survey). The local model has become quite popular recently \cite{warner1965randomized, kasiviswanathan2011can, beimel2008distributed}. The distribution estimation problem considered in this paper has been studied in \cite{warner1965randomized, duchi2013local, erlingsson2014rappor, kairouz2014extremal, wang2016mutual, pastore2016locally, kairouz2016discrete, tianhao2017ldp, acharya2018hadamard,ye2018optimal, bassily2019linear, chen2020breaking}.
Among them, \cite{wang2016mutual,tianhao2017ldp, ye2018optimal,acharya2018hadamard} have achieved worst-case optimal sample complexity over all privacy regime. \cite{chen2020breaking} provides a tight characterization on the trade-off between estimation accuracy and sample size under fixed  privacy and communication constraints.  Their results are tight for all privacy regimes, but have not considered sparsity. 
\citet{kairouz2016discrete} propose a heuristic called projected decoder, which empirically improves the utility for estimating skewed distributions. They also propose a method to deal with open alphabets, which also reduces the dimensionality of the original distribution first by using hash functions. However, hash functions are not invertible, so they use least squares to recover the original distribution, which has no theoretical guarantee on the estimation error even for sparse distributions. 
Recently, \cite{acharya2021estimating} studied the same problem as in this work. Their method combined one-bit Hadamard response with sparse projection onto the probability simplex. Their methods have provable theoretical guarantees but there are some technical limitations. {They also proposed a method, which combined sparse projection with RAPPOR \cite{erlingsson2014rappor}. It achieves optimal sample complexity, but the communication complexity of RAPPOR is $O(k)$ bits for each user, where $k$ is the domain size of the distribution. Recently, \cite{feldman2021lossless} lowered the communication complexity of RAPPOR to $O(\log k)$ by employing pseudo random generator. This result is still worse than ours, which only requires $1$ bit for each user.} The heavy hitter problem, which is closely related to distribution estimation, is also extensively studied in the local privacy model \cite{bassily2017practical, bun2018heavy, wang2019locally, acharya2019communication}.  \cite{wang2019sparse} studies $1$-sparse linear regression under LDP constraints.  Statistical mean estimation with sparse mean vector is also studied under local privacy, e.g. \cite{duchi2019lower, barnes2020fisher}.

\subsection{Preliminaries on Compressive Sensing}\label{sec:preliminary}
Let $x$ be an unknown $k$-dimensional vector. The goal of compressive sensing (CS) is to reconstruct $x$ from only a few linear measurements \cite{candes2005decoding, candes2006robust, donoho2006compressed}. To be precise, let $B \in \mathbb{R}^{m\times k}$ be the measurement matrix with $m \ll k$ and $e \in R^{m}$ be an unknown noise vector, given $y = Bx + e$, CS aims to recover a sparse approximation $\hat{x}$ of $x$ from $y$. This problem is ill-defined in general, but when $B$ satisfies some additional properties, it becomes possible \cite{candes2005decoding,donoho2006compressed}. In particular, the \emph{Restricted Isometry Property} (RIP) is widely used.
\begin{definition}[RIP]\label{def:RIP}
	The matrix $B$ satisfies $(s, \delta)$-RIP property if for every $s$-sparse vector $x$, 
	\begin{align*}
		(1-\delta)\norm{x}_2  \leq \norm{Bx}_2 \leq (1+\delta)\norm{x}_2.
	\end{align*}
\end{definition}
We will use the following results from \cite{cai2013sparse}
\begin{lemma}
	\label{lemma:cs_recovery}
	If $B$ satisfies $(2s, 1/\sqrt{2})$-RIP. Given $y=Bx+e$, there is a polynomial time algorithm, which outputs $\hat{x}$ that satisfies $\norm{x - \hat{x}}_2 \leq \frac{C}{\sqrt{s}}\|x-[x]_s\|_1+ D \|e\|_2$ for some constant $C,D$.
\end{lemma}
For the medium privacy regime, we will use the notion of hierarchical sparsity and a model-based RIP condition \cite{roth2016reliable, roth2018hierarchical}.
\begin{definition}[Hierarchical sparsity \cite{roth2018hierarchical}]
	\label{def:hsparsity}
	Let $x$ be a $k_1 k_2$-dimensional vector  consists of $k_1$ blocks, each of size $k_2$ (e.g., $x=[x^{(1)},\cdots,x^{(k_1)}]$). Then $x$ is $(s, \sigma)$-hierarchically sparse if at most $s$ blocks have non-zero entries and each of these blocks is $\sigma$-sparse.
\end{definition}
\begin{definition}[HiRIP \cite{roth2018hierarchical}]
	\label{def:hirip}
	A matrix $A \in \mathbb{R}^{m \times k}$, where $k = k_1\times k_2$, is $(s,\sigma)$-HiRIP with constant $\delta$ if for all $(s, \sigma)$-hierarchically sparse vectors $x \in \mathbb{R}^{k_1k_2}$, we have
	\begin{align}
		\label{eq:hirip}
		(1 - \delta)\|x\|_2 \leq \|A x\|_2 \leq (1 + \delta) \|x\|_2.
	\end{align}
	
\end{definition}
\begin{theorem}[HiRIP of Kronecker product \cite{roth2018hierarchical}]
	\label{thm:kronecker_hrip}
	For any matrix $A \in \mathbb{R}^{M \times K}$  that is $(s, \delta_1)$ -RIP and any matrix $B \in \mathbb{R}^{m\times k }$ that is $(\sigma, \delta_2)$-RIP, $A\otimes B\in\mathbb{R}^{Mm \times Kk}$ satisfies $(s, \sigma)$-HiRIP with constant
		$\delta_{s, \sigma} \leq \delta_1 + \delta_2 + \delta_1\delta_2$.
\end{theorem}
\begin{theorem}[Recovery guarantee for hierarchically sparse vectors \cite{roth2016reliable}]
	\label{thm:hirip_recovery}
	Suppose matrix $A \in \mathbb{R}^{m \times k}$ ($k = k_1\times k_2$) is $(3s, 2\sigma)$-HiRIP with constant $\frac{1}{\sqrt{3}}$. Given $y = A x + e$ where $x \in \mathbb{R}^{k}$ is $(s, \sigma)$-hierarchically sparse, there is a polynomial time algorithm, which outputs $\hat{x}$ satisfying $\|x - \hat{x}\|_2 \leq C \|e\|_2$ for some constant $C$.
\end{theorem}
\section{One-Bit Compressive Privatization}
To estimate sparse distributions, Acharya et al.~\cite{acharya2021estimating} simply add a sparse projection operation at the end of the one-bit HR scheme proposed in \cite{acharya2019communication}. 
\para{One-bit HR.} The users are partitioned into $K$ groups of the same size deterministically, with $K$ being the smallest power of $2$ larger than $k$. Let $S_1,\cdots, S_K$ be the groups. Since the partition can be arbitrary, we assume $S_j:=\{i\in[n]~|~ i\equiv j ~\mathrm{mod }~ K\}$. Let $H_K$ be the $K \times K$ Hadamard matrix and $H_{i,j}$ be the $(i,j)$-entry. In one-bit HR, each user $i$ in group $S_j$ with a sample $X_i$ sends a bit $Y_i\in\{0,1\}$ distributed as
 \begin{align}
        \operatorname{Pr}\left(Y_{i}=1\right)=\left\{\begin{array}{ll}
\frac{e^{\varepsilon}}{e^{\varepsilon}+1}, & H_{j, X_{i}} = 1, \\
\frac{1}{e^{\varepsilon}+1}, & H_{j,X_{i}} = -1.
\end{array}\right.
    \end{align}
Let $t_j:= P(Y_i = 1| i \in S_j)$ for $j \in [K]$ and $\mathbf{t} := (t_1, \cdots, t_K)$. The key observation is that
\begin{align}
\label{eq:one_bit_hr}
    \frac{e^\eps + 1}{e^\eps - 1}(2\mathbf{t} - \mathbf{1_K}) = H_K \cdot p.
\end{align}
Let $\hat{\mathbf{t}} := (\hat{t}_1,\cdots, \hat{t}_K)$ where $\hat{t}_j := \frac{1}{|S_j|}\sum_{i \in S_j} Y_i$ is the fraction of messages from $S_j$ that are $1$. Then $\hat{\mathbf{t}}$ is an unbiased empirical estimator of $\mathbf{t}$; and $\hat{p} = \frac{e^\eps + 1}{K(e^\eps - 1)}H^T_K(2\mathbf{t} - \mathbf{1_K})$ is an unbiased estimate of $p$ since $\frac{1}{K} H_K^T H_K = I$. 

We note that to make the above estimation process well-defined, the number of samples $n$ must be larger than $K$, since otherwise some group $S_j$ will be empty and the corresponding $\hat{t}_j$ is undefined. Moreover, the proof of \cite{acharya2021estimating} relies on the fact that each $\hat{t}_j$ is the average of $|S_j|$ i.i.d.\ Bernoulli random variables, which implies $\hat{t}_j -t_j$ is sub-Gaussian with variance $\frac{1}{|S_j|}$. However, if the group $S_j$ is empty, this doesn't hold anymore.

\begin{algorithm}
	\SetAlgoLined
	\KwResult{$\hat{p} \in \Delta_k$: an estimate of $p$}
	\KwInput{$X_1, \cdots X_n$ i.i.d from $p$, privacy parameter $\varepsilon$, sparsity $s$, measurement matrix $A \in \mathbb{R}^{m \times k}$}
	For $x \in [m]$, let $B_x := \{y \in [k]: A(x, y) = 1\}$ be the columns where the $x$th row has 1.\\
	Divide the $n$ users into $m$ sets $S_1, \cdots, S_m$ deterministically by assigning all $i\equiv j$ mod $m$ to $S_j$ for $i \in [n]$.\\
	$\forall j \in [m]$ and $\forall i \in S_j$, the distribution of the one-bit message $Y_i$ is
    \begin{align}
        \operatorname{Pr}\left(Y_{i}=1\right)=\left\{\begin{array}{ll}
\frac{e^{\varepsilon}}{e^{\varepsilon}+1}, & X_{i} \in B_{j} \\
\frac{1}{e^{\varepsilon}+1}, & \text { otherwise }
\end{array}\right.
    \end{align}
	Let $\hat{\mathbf{t}} := (\hat{t}_1,\cdots, \hat{t}_m)$ where $\forall j \in [m]$, $\hat{t}_j := \frac{1}{|S_j|}\sum_{i \in S_j} Y_i$ is the fraction of messages from $S_j$ that are $1$.\\
	Apply Lemma 1.3, with $y = \frac{e^\eps + 1}{\sqrt{m}(e^\eps - 1)}(2\hat{\mathbf{t}} - \mathbf{1}_m)$, $B = \frac{1}{\sqrt{m}}A$ and sparsity $s$; let $\tilde{p}$ be the output\\
	Compute the projection of $\tilde{p}$ onto $\Delta_k$, denoted as $\hat{p}$.
	\caption{1-bit Compressive privatization}
\label{algorithm:one_bit_cp}
\end{algorithm}

\para{One bit compressive privatization.}  To resolve the above issue, our scheme doesn't apply one-bit HR but a variant of it. The intuition of our compressive privatization mechanism is that when the distributions are restricted to be sparse, by the theory of compressive sensing, one can use far fewer linear measurements to recovery a sparse vector. In our CP method (shown in Algorithm \ref{algorithm:one_bit_cp}), we do not require the response matrix to be invertible as the Hadamard matrix used in one-bit HR. Any matrix $A \in \{-1,+1\}^{m\times k}$ that satisfies the RIP condition will suffice. More specifically, given target sparsity $s$, we require $\frac{1}{\sqrt{m}}A$ to satisfy $(s, 1/\sqrt{2})$-RIP.
The privatization scheme for each user is almost the same as in one-bit HR, with the Hadamard matrix being replaced by the matrix $A$ above; and clearly this also satisfies $\eps$-LDP. In this case, the relation between $\mathbf{t}$ and $p$ in \eqref{eq:one_bit_hr} now becomes to
\begin{align}
    \label{eq:one_bit_cp}
    \frac{e^\eps + 1}{e^\eps - 1}(2\mathbf{t} - \mathbf{1_m}) = A\cdot p.
\end{align}
On the server side, since $A$ is not necessarily invertible, we need to use sparse recovery algorithms to estimate $p$. More specifically, we reformulate \eqref{eq:one_bit_cp} as
\begin{align}
\label{eq:one_bit_cp_recovery}
    \frac{e^\eps + 1}{\sqrt{m}(e^\eps - 1)}(2\hat{\mathbf{t}} - \mathbf{1}_m) = \frac{1}{\sqrt{m}} A \cdot p +  \frac{2(e^\eps + 1)}{\sqrt{m}(e^\eps - 1)}(\hat{\mathbf{t}} - \mathbf{t}).
\end{align}
Then we can directly apply Lemma \ref{lemma:cs_recovery} to compute a sparse vector $\hat{p}$, with $l_2$ error, i.e. $\|p-\hat{p}\|_2$, proportional to $\frac{2(e^\eps + 1)}{\sqrt{m}(e^\eps - 1)}\|\hat{\mathbf{t}} - \mathbf{t}\|_2$. Note $\mathbb{E}[\|\hat{\mathbf{t}} - \mathbf{t}\|_2]$ is the MSE of the empirical estimator $\hat{\mathbf{t}}$, which only depends on $m$ rather than on $k$. Moreover, Lemma \ref{lemma:cs_recovery} can handle approximately sparse vectors, and thus the above argument is immediately applicable to the setting when $p$ is only approximately sparse. 
The result is summarized in the following theorem.
\begin{theorem}[high privacy regime]
\label{thm:one_bit_cp_error_bound}
Given any matrix $A \in \{\pm1\}^{m\times k}$ with $\frac{1}{\sqrt{m}}A$ satisfies $(s, {1}/{\sqrt{2}})$-RIP, assume $\eps=O(1)$ and $p$ is $s$-sparse, then for a target error $\alpha$, the sample complexity of our method is $O(\frac{m}{\eps^2\alpha^2})$ for $\ell_2$ error and  $O(\frac{sm}{\eps^2 \alpha^2})$ for $\ell_1$ error. The $\ell_2$ result also holds for $(s, \sqrt{s}\alpha)$-sparse $p$ and the $\ell_1$ result also holds for $(s,\alpha)$-sparse $p$. The communication cost for each user is 1 bit.
\end{theorem}
\begin{proof}
We first consider the case for $\ell_2$ error. Since $\Delta_k$ is convex, $\|\hat{p} - p\|_2 \leq \| \tilde{p} - p\|_2$. By Lemma 1.3, we know:
\begin{align}
\label{eq:cp_recovery_error}
    \E\left[\|p - \tilde{p}\|_2\right] \leq \frac{C}{\sqrt{s}}\|p - [p]_s\|_1 + D\cdot\frac{2(e^\eps + 1)}{\sqrt{m}(e^\eps - 1)}\E\left[  \|\hat{\mathbf{t}} - \mathbf{t}\|_2\right],
\end{align}
where $C$ and $D$ are absolute constants. Since $\hat{\mathbf{t}}$ is an empirical estimator of $\mathbf{t}$, we have
\begin{align}
\label{eq:t_estimation_error}
\E^2\left[\|\hat{\mathbf{t}} - \mathbf{t}\|_2\right]\leq \E\left[\|\hat{\mathbf{t}} - \mathbf{t}\|_2^2\right]  = \sum_{y = 1}^m \frac{1}{|S_y|^2}\sum_{j \in S_y}\var(Y_j) \leq \frac{m^2}{4n},
\end{align}
where the first inequality is from Jensen's inequality and the last inequality is from that $Y_j \in \{0, 1\}$. Combining \eqref{eq:cp_recovery_error} and \eqref{eq:t_estimation_error} yields that, 
\begin{align}
\label{eq:cp_l2_uppper_bound}
    \E\left[\|p - \tilde{p}\|_2\right] \leq\frac{C}{\sqrt{s}}\|p - [p]_s\|_1 + D\cdot \frac{e^\eps + 1}{e^\eps  - 1}\sqrt{\frac{m}{n}}.
\end{align}
When $p$ is $s$ sparse, the first term in \eqref{eq:cp_l2_uppper_bound} is 0. Thus, when $n \geq \frac{b m}{\alpha^2 \eps^2}$ for some large enough constant $b$, the expected $\ell_2$ error is at most $\alpha$. For $(s, \sqrt{s}\alpha)$-sparse $p$, the first error term in \eqref{eq:cp_l2_uppper_bound} is bounded by $O(\alpha)$, thus the result still holds for approximately sparse case. For $\ell_1$ error, we use $L_1$ projection in the final step, which means projection by minimizing $L_1$ distance. In this way, when $p$ is $s$ sparse, we have 
\begin{align}
\label{eq:l1_projection}
    \|p - \hat{p}\|_1 \leq \|p - \tilde{p}\|_1 + \|\tilde{p} - \hat{p}\|_1 \leq 2  \|\tilde{p} - p\|_1 \leq 2\sqrt{2s}\|p - \tilde{p}\|_2
\end{align}
where the first inequality is from triangle inequality, the second inequality is from the $L_1$ projection and the last inequality is from Cauchy-Schwartz and the fact that $\tilde{p} - p$ is $2s$-sparse. Thus to achieve an $\ell_1$ error of $\alpha$, it's sufficient to get an estimate with $\alpha' = \alpha / \sqrt{s}$ error for $\ell_2$. The sample complexity is $O(sm/\eps^2\alpha^2)$. For $\ell_1$ error with $p$ being $(s, \alpha)$-sparse, now $p - \tilde{p}$ is $(2s, \alpha)$-sparse. We have
\begin{align*}
     \|p-\tilde{p}\|_1 &=\|[p-\tilde{p}]_{2s}\|_1 + \|(p-\tilde{p}) - [p-\tilde{p}]_{2s}\|_1 \\&\le \sqrt{2s}\|p-\tilde{p}\|_2 + \alpha.
\end{align*}
Then, the $\ell_1$ result follows by a similar argument as for the exact sparse case.
\end{proof}

\subsection{Guarantees on Random Matrices}
The measurement matrix we use is $B=\frac{1}{\sqrt{m}} A$, where the entries of $A\in \{-1,+1\}^{m\times k}$ are i.i.d. Rademacher random variables, i.e., takes $+1$ or $-1$ with equal probability. It is known that for $m \ge O(s\log{\frac{k}{s}})$, $B$ satisfies $(s, 1/\sqrt{2})$-RIP with probability $1 - e^{-m}$ \cite{baraniuk2008simple}. By Theorem \ref{thm:one_bit_cp_error_bound}, we have the following theorem. 
\begin{theorem}\label{thm:RandomMain_onebit}
	For $m=O\left(s\log \frac{k}{s} \right)$ and $\eps = O(1)$, if the entries of $A$ are i.i.d.\ Rademacher random variables, then with probability at least $1-e^{-m}$, our method has sample complexity $O\left(\frac{s 
	\log{(k / s)}}{\eps^2\alpha^2}\right)$ for $\ell_2$ error, and $O\left(\frac{s^2\log{(k/s)}}{\eps^2\alpha^2}\right)$ for for $\ell_1$ error. The $\ell_2$ result holds for $(s, \sqrt{s}\alpha)$-sparse $p$ and the $\ell_1$ result holds for $(s,\alpha)$-sparse $p$. The communication cost for each user is 1 bit.
\end{theorem}  
\para{Lower bound.} \cite{acharya2021estimating} proves a lower bound of $n = \Omega(\frac{s^2\log{(k/s)}}{\eps^2 \alpha^2})$ on the sample complexity for $\ell_1$ error with $\eps = O(1)$. This matches our bound up to a constant.
\section{Symmetric Compressive Privatization}\label{sec:symmetric}
The one-bit compressive privatization scheme is asymmetric, where users in different groups apply different privatization schemes. In this section, we introduce a symmetric version of compressive privatization, which could be easier to implement in real applications. 

To estimate an unknown distribution $p \in \mathbb{R}^k$, all previous symmetric LDP mechanisms essentially apply a probability transition matrix $Q$ mapping $p$ to $q$, where $q$ is the distribution of the privatized samples. The central server get $n$ independent samples from $q$, from which it computes an empirical estimate of $q$, denoted as $\hat{q}$, and then computes an estimator of $p$ from $\hat{q}$ by solving $Qp = \hat{q}$. The key is to design an appropriate $Q$ such that it satisfies privacy guarantees and achieves low recovery error. The error of $\hat{p} = Q^{-1}p$ is dictated by the estimation error of $\hat{q}$ and the spectral norm of $Q^{-1}$. In our symmetric scheme, we map $p$ to a much lower dimensional $q$, and then $\hat{q}$ with similar estimation error can be obtained with \emph{much less number of samples}. However, now $Q$ is not invertible; to reconstruct $\hat{p}$ from $\hat{q}$, we use sparse recovery \cite{candes2005decoding, candes2006stable}. 
\para{Privatization.} Our mechanism $Q$ is a mapping from $[k]$ to $[m]$. For each $x\in [k]$, we pick a set $C_x \subseteq [m]$, which will be specified later, and let $n_x= |C_x|$. Our privatization scheme $Q$ is given by the conditional probability of $y$ given $x$:
\begin{align}
\label{equation:Q}
Q(y |x):=\left\{\begin{array}{ll}
\frac{e^{\varepsilon}}{n_x e^{\varepsilon}+m-n_x} & \text { if } y \in C_{x}, \\
\frac{1}{n_x e^{\varepsilon}+m-n_x} & \text { if } y \in [m] \backslash C_{x}.
\end{array}\right.
\end{align}
Note $Q(y |x)$ is an $m$-ary distribution for any $x\in[k]$.  For each $x\in [m]$, we define its incidence vector as $I_x\in \{-1,+1\}^m$ such that $I_{x}(j) =+1$ iff $j\in C_x$. Let $A \in \mathbb{R}^{m \times k}$ be the matrix whose $x$-th column is $I_x$. Each user $i$ with a sample $X_i$ generates $Y_i$ according to (\ref{equation:Q}) and then send $Y_i$ to the server with communication cost $O(\log m)$ bits.
\para{Sufficient conditions for $Q$.} The difference between our mechanism, RR \cite{warner1965randomized} and HR \cite{acharya2018hadamard} is the choice of each $C_x$, or equivalently the matrix $A$. In RR, $C_x =\{x\}$ for all $x\in[k]$, while in HR, $A$ is the Hadamard matrix.  In our privatization method, any matrix $A$ whose column sums are close to $0$ and that satisfies RIP will suffice. More formally, given target error $\alpha$, privacy parameter $\eps$ and target sparsity $s$, we require $A$ to have the following $2$ properties:
\begin{itemize}
	\item \textbf{P1:} $(1-\beta)\frac{m}{2} \leq n_i \leq (1 + \beta)\frac{m}{2}$ for all $i\in [k]$, with $\beta\le \frac{\eps}{2}$ and $\beta \le {c \alpha}$ for some $c$ depending on the error norm.
	\item \textbf{P2:} $\frac{1}{\sqrt{m}}A$ satisfies $(s, \delta)$-RIP, where $s$ is the target sparsity and $\delta\le 1/\sqrt{2}$.
\end{itemize}
In \cite{acharya2018hadamard}, Hadamard matrix is used to specify each set $C_x$. The proportion of $+1$ entries is exactly half in each column of $H$ (except for the first column), and thus P1 is automatically satisfied with $\beta = 0$. Since Hadamard matrix is othornormal, it is $(k,0)$-RIP.
\subsection{Estimation Algorithm}
We first show how to model the estimation of $p$ as a standard compressive sensing problem. Recall $n_i$ is the number of $+1$'s in the $i$th column of $A$. Let $q\in\Delta_m$ be the distribution of a privatized sample, given that the input sample is distributed according to $p\in \Delta_k$. Then, for each $j\in [m]$, we have

\begin{align*}
q_j &= \sum_{i = 1}^{k}p_i\cdot Q(Y= j |X = i) \\ &= \sum_{i: j\in C_i} \frac{e^{\varepsilon} \cdot p_i}{n_i e^{\varepsilon}+m-n_i}  + \sum_{i: j \in [k]\backslash C_i} \frac{ p_i}{n_i e^{\varepsilon}+m-n_i}.
\end{align*}

By writing the above formula in the matrix form, we get 
\begin{align}
\label{equation:p_to_q}
q = \left(\frac{e^\varepsilon - 1}{2}A + \frac{e^\varepsilon + 1}{2}J\right) D p,
\end{align}
where $J \in \mathbb{R}^{m \times k}$ is the all-one matrix and $D$ is the diagonal matrix with $d_i = \frac{1}{n_i e^\varepsilon +m - n_i}$ in the $i$th diagonal entry. 
As mentioned above, the matrix $\frac{1}{\sqrt{m}} A$ to be used will satisfy RIP. We then rewrite (\ref{equation:p_to_q}) to the form of a standard noisy compressive sensing problem
\begin{align*}
\underbrace{\frac{(e^\varepsilon + 1)}{(e^\varepsilon - 1)} \left(\sqrt{m}q -\frac{\mathbf{1}}{\sqrt{m}} \right)}_y \nonumber= \underbrace{\frac{1}{\sqrt{m}}A}_B \underbrace{ (D' p)}_x  + \underbrace{\frac{e^\varepsilon + 1}{\sqrt{m}(e^\varepsilon - 1)}J (D' - I)  p}_e.
\end{align*}
where $\mathbf{1} \in \mathbb{R}^m$ is the all-one vector, $I \in \mathbb{R}^{k \times k}$ is the identity matrix and $D' = \frac{m(e^\varepsilon + 1)}{2} D$. 
The exact $q$ is also unknown, and we can only get an empirical estimate $\hat{q}$ from the privatized samples. So, we need to add a new noise term that corresponds to the estimation error of $\hat{q}$, and the actual under-determined linear system is
\begin{align}
\label{eq:p_to_q_cs_final}
\frac{(e^\varepsilon + 1)}{(e^\varepsilon - 1)} \left(\sqrt{m}\hat{q} - \frac{\mathbf{1}}{\sqrt{m}}\right) &= \frac{1}{\sqrt{m}}A D' p  +  \underbrace{\frac{e^\varepsilon  + 1}{\sqrt{m}(e^\varepsilon - 1)}J(D' - I) p}_{e_1} \notag\\ &+ \underbrace{\frac{\sqrt{m}(e^\varepsilon + 1)}{e^\varepsilon - 1}(\hat{q} - q)}_{e_2}.\end{align}
Given the LHS of \eqref{eq:p_to_q_cs_final}, $\frac{1}{\sqrt{m}}A$ and a target sparsity $s$, we reconstruct $\widehat{D'p}$ by applying Lemma~\ref{lemma:cs_recovery}. Then compute $\hat{p}' = D'^{-1}\widehat{D'p}$ ($D'$ is known) and project it to the probability simplex. The pseudo code of the algorithm is presented in Algorithm \ref{algorithm:scp}.
\begin{algorithm}[h]
	\SetAlgoLined
	\KwResult{$\hat{p} \in \Delta_k$: an estimate of $p$}
	\KwInput{$X_1, \cdots, X_n$ i.i.d from $p$, privacy parameter $\eps$, sparsity $s$, $A \in \mathbb{R}^{m \times k}$}
	$\forall x \in [k]$, Let $C_x := \{y: A(y, x) = 1\}$. Let $n_x$ be the number of $+1$ in the $x$-th column of $A$. Then for each $X_i, i \in [n]$, the  privatized sample $Y_i$ is generated  according to the following distribution
	\begin{align*}
	\operatorname{Pr}[Y_i = y | X_i = x] = \left\{\begin{array}{cc}
	\frac{e^\eps}{n_x e^\eps + m - n_x}   &  y \in C_x,\\
	\frac{1}{n_x e^\eps + m - n_x} &  \mbox{otherwise}
	\end{array}\right.
	\end{align*}\\
	$\hat{q}$ = $(0, 0, \cdots 0) \in \mathbb{R}^m$, 	$\mathbf{1} = (1,1,\cdots, 1) \in \mathbb{R}^{m}$ \\
	\For{$i \gets 1$ \KwTo $m$}{
		$\hat{q}[i]$ = $\frac{\sum_{j=1}^{n}\mathbb{I}(Y_j = i)}{n}$
	}
	Apply Lemma 1.3, with $y=\frac{(e^\varepsilon + 1)}{(e^\varepsilon - 1)} \left(\sqrt{m}\hat{q} - \mathbf{1}\right)$, $B=\frac{1}{\sqrt{m}}A$ and sparsity $s$; let $f$ be the output \\
	Compute $\hat{p}' = D'^{-1}f$, and then compute the projection of $\hat{p}'$ onto the set $\Delta_k$, denoted as $\hat{p}$\\
	\Return $\hat{p}$
	\caption{Symmetric compressive privatization}
	\label{algorithm:scp}
\end{algorithm}
\subsection{Privacy Guarantee and Sample Complexity}
In this section, we will provide the privacy guarantee and sample complexity of our symmetric privatization scheme. All the proofs are in the supplementary material.
\begin{lemma}[Privacy Guarantee]
\label{lemma:scp_privacy}
	If $(1-\beta)\frac{m}{2} \leq n_i \leq (1 + \beta)\frac{m}{2}$ for all $i \in [k]$ and $0 \leq  \beta <1$, then the privacy mechanism $Q$ from \eqref{equation:Q} satisfies $(\varepsilon+2\beta)$-LDP.
\end{lemma} 
When the matrix $A$ used in \eqref{equation:Q} satisfies$ (1-\eps)\frac{m}{2} \leq n_i \leq (1 + \eps)\frac{m}{2} \textrm{ for all } i \in [k]
$, then $Q$ is $3\eps$-LDP.  We can rescale $\eps$ in the beginning by a constant to ensure $\eps$-LDP, which will only affect the sample complexity by a constant factor. Next we consider the estimation error. By Lemma \ref{lemma:cs_recovery}, the reconstruction error depends on the $\ell_2$ norm of $e_1$ and $e_2$  in \eqref{eq:p_to_q_cs_final}.
\begin{lemma}
\label{lemma:scp_noise_1}
For a fixed $ \beta \in [0,0.5)$, if $n_i \in (1 \pm \beta) \frac{m}{2}$ for all $i\in[k]$, then $\norm{e_1}_{2} \leq \frac{\beta }{1-\beta} \le 2\beta$.
\end{lemma}
\begin{lemma}
\label{lemma:scp_noise_2}
$
\mathbb{E}\left[\|e_2\|_2\right] \leq \frac{e^\varepsilon + 1}{e^\varepsilon -1}\sqrt{\frac{m}{n}}$.
\end{lemma}
Combining Lemma \ref{lemma:scp_noise_1}, \ref{lemma:scp_noise_2} and Lemma \ref{lemma:cs_recovery}, we can bound the estimation error. 
\begin{theorem}[Estimation error]
\label{theorem:scp_reconstruction_error}
If $A$ satisfies two properties in Section \ref{sec:symmetric}, for some constant $C$,
\begin{align*}
    \mathbb{E}\left[\norm{p - \hat{p}'}_2\right] &\leq \left(1 + \frac{\beta}{2}\right)\left(2D\beta+ \frac{D(e^\varepsilon + 1)}{e^\varepsilon -1}\sqrt{\frac{m}{n}} \right)\\&+\left(1+\frac{\beta}{2}\right)^2\left( \frac{C}{\sqrt{s}}\|p-[p]_s\|_1\right).
\end{align*}
\end{theorem}

The sample complexity to achieve an error $\alpha$ and $\eps$-LDP for $0\le \eps\le 1$ is summarized as follows.
\begin{corollary}
\label{corollary:sample complexity_symm}
If $A$ satisfies two properties in section \ref{sec:symmetric} with $\beta\le c\alpha/4D$, with  $c=1$ for $\ell_2$ error and $c=\frac{1}{\sqrt{s}}$ for $\ell_1$ error,  and $p$ is $s$ sparse, the sample complexity of our symmetric compressive privatization scheme is $O(\frac{m}{\epsilon^2\alpha^2})$ for $\ell_2$ error and  $O(\frac{sm}{\epsilon^2 \alpha^2})$ for $\ell_1$ error. The $\ell_2$ result also holds for $(s, \sqrt{s}\alpha)$-sparse $p$ and the $\ell_1$ result also holds for $(s,\alpha)$-sparse $p$. The communication cost for each user is $\log m$ bits.
\end{corollary}
{This corollary can be directly derived by applying a similar proof as that of Theorem \ref{thm:one_bit_cp_error_bound}, so we omit the proof here.} 
\subsection{Guarantees on Random Matrices}
The response matrix we used here is $B = \frac{1}{\sqrt{m}}[A_{\frac{m}{2}}; -A_{\frac{m}{2}}]$,  where $A_{\frac{m}{2}}$ is a $\frac{m}{2} \times {k}$ Rademacher matrix. It can be easily shown that, if $m = \Omega(s\log \frac{k}{s})$, $B$ is $(s, \frac{1}{\sqrt{2}})$-RIP with probability $1 - e^{-m}$. At the same time, $B$ satisfies P1 with $\beta = 0$. Thus by Corollary \ref{corollary:sample complexity_symm}, we get the following result.
\begin{theorem}\label{thm:RandomMain_symm}
	For $m=\Omega(s\log\frac{k}{s})$, if $A$ is defined as above, then with probability at least $1-e^{-m}$, our symmetric compressive privatization scheme has sample complexity $O\left(\frac{m}{\eps^2\alpha^2}\right)$ for $\ell_2$ error, and $O\left(\frac{sm}{\eps^2\alpha^2}\right)$ for for $\ell_1$ error. The $\ell_2$ result holds for $(s, \sqrt{s}\alpha)$-sparse $p$ and the $\ell_1$ result holds for $(s,\alpha)$-sparse $p$. The communication cost for each user is $\log m = O(\log s+\log\log k)$ bits.
\end{theorem} 
\para{Communication lower bound}. Note that the sample complexity is the same as that of one-bit compressive privatization. But the communication complexity is $\log m$ bits. \cite{acharya2019communication} proves that, for symmetric schemes, the communication cost is at least $\log k - 2$ for general distribution estimation. Thus, even if the sparse support is known, the communication cost is at least $\log s -2$ bits. Our symmetric scheme requires $\log s + \log\log k$ bits, which is optimal up to a $\log\log k$ additive term.
\section{Recursive Compressive Privatization for Medium Privacy Regimes}\label{sec:median}
For high privacy regime $\eps = O(1)$, we have provided symmetric and asymmetric privatization schemes that both achieve optimal sample  and communication complexity. For medium privacy regime $1 < e^\eps < k$, the relationship between sample complexity, privacy, and communication cost become more complicated.  Recently, \cite{chen2020breaking} provides a clean characterization for any $\eps$ and communication budget $b$ for dense distributions. Interestingly, they show that the complexity is determined by the more stringent constraint, and the less stringent constraint can be satisfied for free. In this section, we provide an analogous result for sparse distribution estimation for privacy regime where $1 \le e^\eps \le s$. 

Our RCP scheme (shown in Algorithm \ref{algorithm:rcp}) consists of three steps including random permutation, privatization and estimation. Let $X$ be a element sampled from $p$, which is viewed as a one-hot vector. Let $A$ be a measurement matrix and 
$Y = AX$. Since $\E[Y]= A\E[X] = Ap$, we want to get an estimator of $p$ by recovering from the empirical mean of $Y$. In one-bit CP, $A$ is a Rademacher matrix, while in RCP, we use the kronecker product of two RIP matrices. In other words, $A = A_1 \otimes A_2$, where $A_1$ and $A_2$ both satisfy RIP condition. It's known from kronecker compressive sensing \cite{duarte2011kronecker} that $A$ also satisfies $s$-RIP if both $A_1$ and $A_2$ satisfy $s$-RIP. However, $s$-RIP condition is too stringent for $A_2$, since $A_2$ measures each block of $p$ and the sparsity of each block could be much less than $s$ on average. 

We use hierarchical  compressive sensing. To make the distribution $p$ hierarchically sparse (see definition~\ref{def:hsparsity}), we first randomly permute $p$ in the beginning and let $p'$ be the resulting distribution. Thus the sparsity of each block in $p'$ is roughly $s/L$, where $L$ is the number of blocks, and hence $p'$ is nearly $(L, s/L)$-hierarchically sparse.
\begin{algorithm}[ht]
	\SetAlgoLined
	\KwResult{$\hat{p} \in \Delta_k$: an estimate of $p$}
	\KwInput{$X_1, \cdots X_n$ i.i.d.\ samples from $p$, privacy parameter $\varepsilon$, sparsity $s$, matrix $A_1 \in \mathbb{R}^{L \times L}$, $A_2 \in \mathbb{R}^{m \times \frac{K}{L}}$ where $K = 2^{\lceil \log_2 k \rceil}$ and $L = \min\{2^b, 2^{\lceil\log_2 e^\eps\rceil}\}$, public random permutation matrix $P \in \mathbb{R}^{K \times K}$. (We represent $X_1, \cdots, X_n$ as one-hot vectors.)}
	Divide the $n$ users into $m$ groups $S_1, \cdots, S_m$: $S_j:=\{i\in[n]~|~ i\equiv j ~\mathrm{mod }~ m\}$.\\
	$\forall j \in [m]$ and $\forall i \in S_j$, pad $(K - k)$ zeroes to the end of $X_i$ and get permuted sample $X_i' = PX_i$. \\
	Define $Q_j(X_i') = [(A_2)_j\cdot X_i'^{(1)}, \cdots, (A_2)_j\cdot X_i'^{(L)}] \in \{-1, 0, 1\}^{L}$. The privatized output $\hat{Q}_j(X_i')$ is defined as follows
	\begin{align}
	\label{eq:rcp_dist}
	\hat{Q}_j(X_i')=\left\{\begin{array}{ll}
	Q_j(X_i'), & \mbox{w.p. $\frac{e^\eps}{e^\eps + 2L - 1}$} \\
	Q' \in \mathcal{Q}\backslash \{Q_j(X_i')\}, & \mbox{w.p. $\frac{1}{e^\eps + 2L - 1}$}
	\end{array}\right.
	\end{align}
	where $\mathcal{Q} = \{\pm e_1, \pm e_2, \cdots, \pm e_L\}$ is the collection of $2L$ standard basis vectors. \\
	For each $j' \in [mL]$ such that $j'\equiv j$ (mod $m$) and $j' = j + (t - 1) \cdot m$, the server computes 
	\begin{align}
	\hat{q}_{j'} = \frac{m}{n}\left(\frac{e^\eps + 2L - 1}{e^\eps - 1}\right)\sum_{i \in S_j} (A_1)_t \cdot \hat{Q}_j(X_i').
	\end{align}\\
	Let $\hat{q} := (\hat{q}_1,\cdots, \hat{q}_{mL})$. Apply Theorem \ref{thm:hirip_recovery} with $y = \frac{1}{\sqrt{mL}}\hat{q}$, measurement matrix $\frac{1}{\sqrt{mL}}(A_1 \otimes A_2)$ and hierarchical sparsity $(L, (1+\beta)\frac{s}{L})$; let $\hat{p}'$ be the output. Let $\tilde{p}$ be the first $k$ elements of $P^{-1}\hat{p}'$.\\
	Compute the projection of $\tilde{p}$ onto $\Delta_k$, denoted as $\hat{p}$.
	\caption{Recursive Compressive Privatization}
	\label{algorithm:rcp}
\end{algorithm}
\para{Hierarchical sparsity after random permutation.} Let $P$ be a random permutation matrix, which is public information. Each user $i$ with sample $X_i$ first compute $X_i' = PX_i$. So $p' = Pp$. We divide $p'$ into $L$ consecutive blocks, then each of the $L$ blocks of $p'$ has sparsity around $\frac{s}{L}$ with high probability. Let $\mathcal{E}$ be the event that
\begin{align}
\label{eq:event_E}
     s_i \leq (1 + \beta)\frac{s}{L},  \textrm{ for all } i \in [L], \mbox{for some $\beta>0$},
\end{align}
where $s_t$ is the sparsity of the $t$th block in $p'$. For notation convenience, we simply use $X_i$ to denote the permuted one-hot sample vector of user $i$. By concentration inequalities,  $\mathcal{E}$ happens with high probability. One twist here is that we cannot apply standard Chernoff-Hoeffding bound, since the sparsity in each block is not a sum of i.i.d.\ random variables. But random permutation random variables are known to be negatively associated (see e.g.\ \cite{wajc2017negative}), so the concentration bounds still holds. We have the following lemma, the proof of which is provided in the supplementary material.
\begin{lemma}\label{lemma:rcp_event_e}
Event $\mathcal{E}$ holds with probability at least $1 - Le^{-\frac{\beta^2s}{2L}}$, and when $\mathcal{E}$ happens, $p'$ is $(L, (1 + \beta)\frac{s}{L})$-hierarchically sparse. 
\end{lemma}
\para{Privatization.} 
In our privatization step, the response matrix is of the form $A = A_1 \otimes A_2$, where $A_2 \in \mathbb{R}^{m \times \frac{k}{L}} $ and $A_1 \in \mathbb{R}^{L \times L}$ are two $\pm 1$ matrices. User $i$ will send a privatized version of $Y_i=AX_i$.  Let $a_{ij}$ be the $(i,j)$-entry of $A_1$, then
\begin{align*}
    Y_i = AX_i =\left[\begin{array}{ccc}
a_{11} {A_2} & \cdots & a_{1 L} {A_2} \\
\vdots & \ddots & \vdots \\
a_{L1} {A_2} & \cdots & a_{L L} {A_2}
\end{array}\right]
\left[\begin{array}{c}
X_i^{(1)} \\
\vdots \\
X_i^{(L)}
\end{array}\right],
\end{align*}
where the one-hot sample vector $X_i$ is divided into $L$ blocks. Since $X_i$ has only one non-zero, the vector $Y_i$ can be encoded with $ m + \log L$ bits, where $\log L$ bits is to specified the block index $\ell$ that contains the non-zero of $X_i$ and $m$ bits is for $A_2X_i^{(\ell)}$. However, this is still too large, so user will only pick one bit from $A_2X_i^{(\ell)}$, which is the $j$th bit if $i\in S_j$. Then the user privatize the $1+\log L$ bits using RR mechanism \cite{warner1965randomized} with alphabet size $2^{2L}$. More formally, let $Q_j(X_i)=[(A_2)_j \cdot X_i^{(1)}, \cdots, (A_2)_j\cdot X_i^{(L)}] \in \{-1, 0, 1\}^{L}$, which is also a one-hot vector. Let $\ell$ be the block index such that $X_i^{(\ell)} \neq 0$, then the $\ell$th bit in $Q_j(X_i)$  is the only non-zero entry, which has value $(A_2)_j\cdot X_i^{(\ell)}$. Then the user computes a privatization of $Q_j(X_i)$ (see \eqref{eq:rcp_dist} in Algorithm~\ref{algorithm:rcp}) with $2^{2L}$-RR. Clearly the communication cost is $l = \log_2 L + 1$ bits. 
\para{Estimation via hierarchical sparse recovery.} By Definition \ref{def:hsparsity}, $p'$ is $(L, (1 + \beta)\frac{s}{L})$-hierarchical sparse. To recover $p'$, by Theorem \ref{thm:kronecker_hrip} and \ref{thm:hirip_recovery}, $\frac{1}{\sqrt{L}}A_1$ and $\frac{1}{\sqrt{m}}A_2$ are required to satisfy $3L$-RIP and $\frac{2(1+\beta)s}{L}$-RIP condition respectively. Since $A_1$ is square, we can use the Hadamard matrix $H_L$. For $A_2$, we use a Rademacher matrix with number of rows $m = \Theta\left((1 + \beta)\frac{s\log({k}/{(1 + \beta) s})}{L}\right)$, so that $\frac{1}{\sqrt{m}}A_2$ satisfies $\frac{2(1+\beta)s}{L}$-RIP. Let $q = Ap'$, which is equivalent to 
\begin{align}
    \frac{1}{\sqrt{mL}}\hat{q} = \frac{1}{\sqrt{mL}} Ap' + \frac{1}{\sqrt{mL}}(\hat{q} - q).
\end{align}
In the estimation algorithm (step 5 in Algorithm \ref{algorithm:rcp}), we use $\hat{q}$ to recover $p'$. By Theorem \ref{thm:hirip_recovery}, we have
\begin{align}
\label{eq:meidum_estimation_error}
    \E[\|p - \hat{p}\|_2] \leq \E[\|p' - \hat{p}'\|_2] \leq \frac{1}{\sqrt{mL}} \mathbb{E}[\|\hat{q} - q\|_2] 
\end{align}
Thus, the estimation error of $\hat{p}$ is bounded by the error of $\hat{q}$. The following lemma gives an upper bound on the error of $\hat{q}$.
\begin{lemma}
\label{lemma:rcp_q_j_variance}
$\forall j' \in [mL], \mathbb{E}\left[\left(\hat{q}_{j^{\prime}}-q_{j^{\prime}}\right)^{2}\right] \leq \frac{m}{n}\left(\frac{e^{\varepsilon}+2L-1}{e^{\varepsilon}-1}\right)^{2}$.
\end{lemma}
Note that we require $m = C'\cdot \frac{(1 + \beta)s\log\left(k / (1 + \beta)s\right)}{L}$ for some absolute constant $C'$. It can be seen that the error in Lemma~\ref{lemma:rcp_q_j_variance} is minimized when $L = \Theta(e^\eps)$ and decreasing as $L$ increases from $1$ to $\Theta(e^{\eps})$. Note that the communication cost is $\log L+1$. Thus, if we are further given a communication budget $b$, we need to set $L \leq 2^{b-1}$. When $e^\eps < 2^b$, the best $L$ is $e^\eps$, which leads to optimal error. If $e^\eps \ge 2^b$, i.e., communication becomes the more stringent constraint, then set $L=2^b$. In other words, $L = \min\{2^b, e^\eps\}$. Combining  \eqref{eq:meidum_estimation_error} and Lemma \ref{lemma:rcp_q_j_variance}, we have the following results on the sample complexity for medium privacy $1 \le e^\eps \le s$. 
\begin{theorem}
\label{thm: rcp_medium_result}
Given $\eps$ and a communication budget $b$, with $1 \le e^\eps , 2^b \le s$ and let $L = \min\{2^b, e^\eps\}$. For $m = \Theta(\frac{(1 + \beta)s\log(k / (1 + \beta)s)}{L})$, with probability at least $1 - Le^{-\frac{\beta^2s}{2L}} - e^{-m}$, our scheme is $\eps$-LDP and has communication cost $\log L +1$, and the sample complexity for $\ell_2$ error is $O\left((1 + \beta)\frac{s\log({k} /{(1 + \beta)s})}{L\alpha^2}\right)$; for $\ell_1$ error, the sample complexity is $O\left((1 + \beta)\frac{s^2\log({k}/{(1 + \beta)s})}{L\alpha^2}\right)$.
\end{theorem}
\begin{figure*}[!ht]
	\centering
	\begin{subfigure}{0.24\textwidth}
		\centering
		\includegraphics[width=\linewidth]{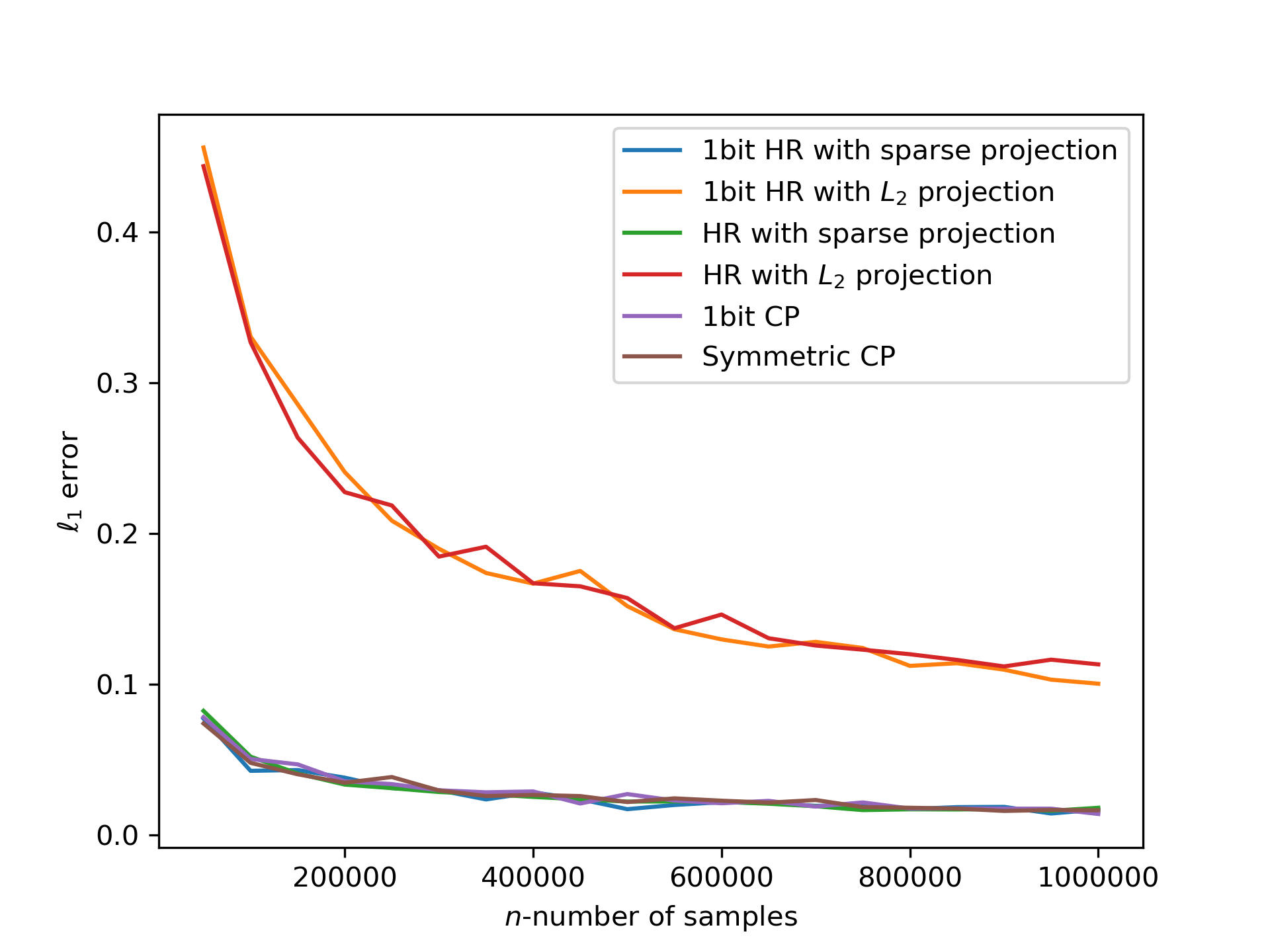}
		\caption{$\mathsf{Unif}(10) $}
		\label{fig:sfig1}
	\end{subfigure}
	\begin{subfigure}{.24\textwidth}
		\centering
		\includegraphics[width=\linewidth]{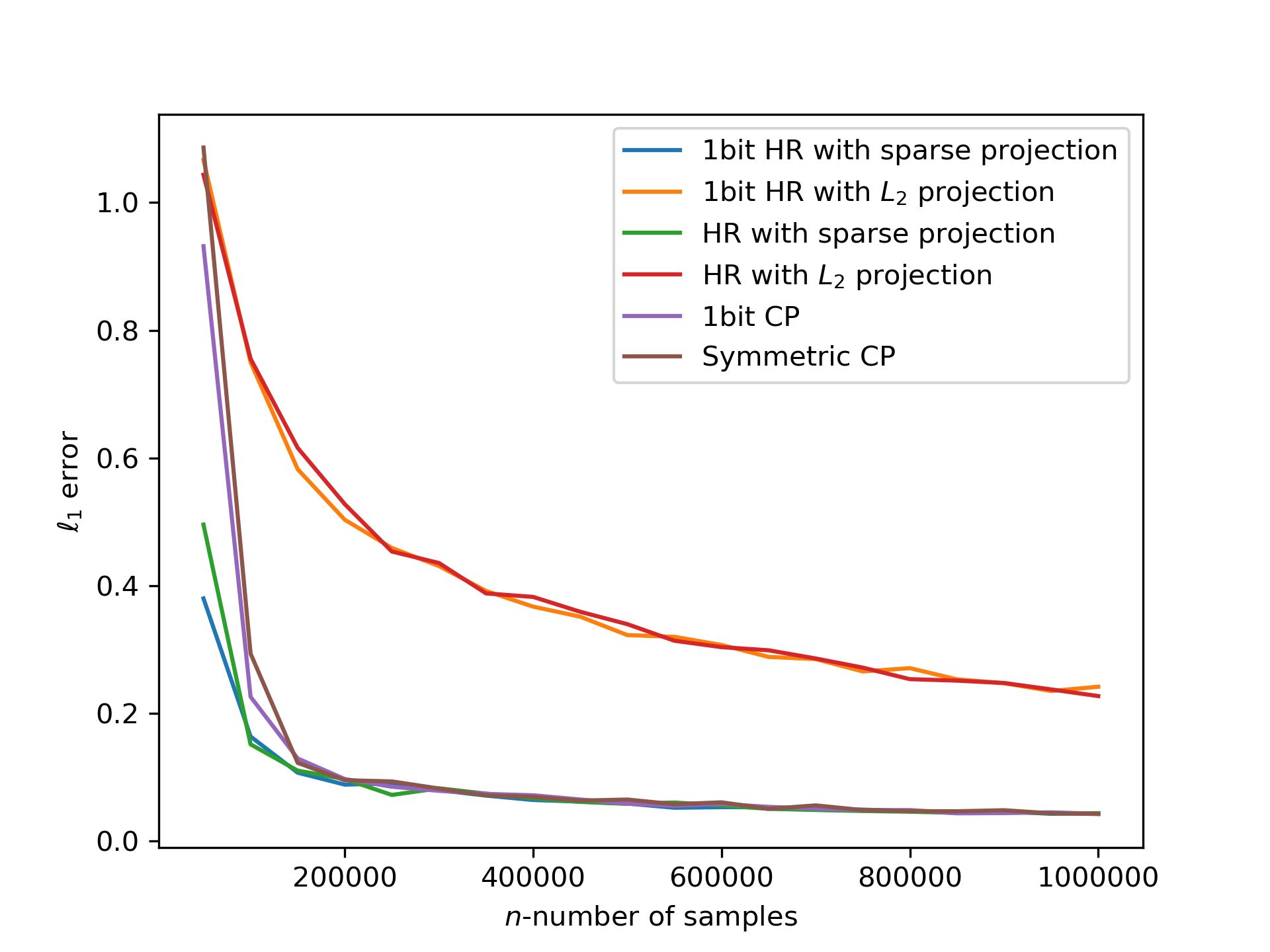}
		\caption{$\mathsf{Unif}(25)$}
		\label{fig:sfig2}
	\end{subfigure}
	\begin{subfigure}{.24\textwidth}
		\centering
		\includegraphics[width=\linewidth]{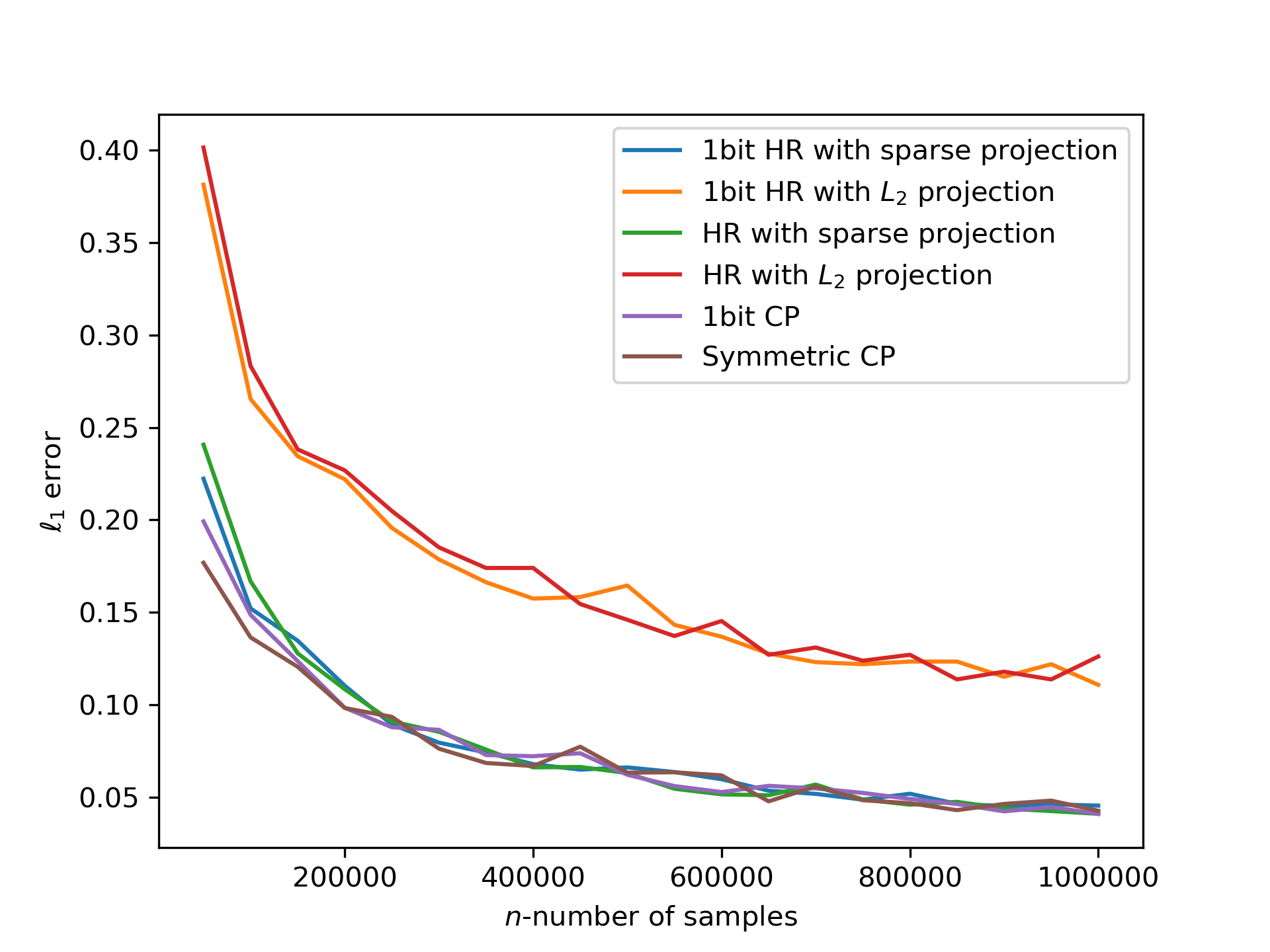}
		\caption{$\mathsf{Geo}(0.6)$}
		\label{fig:sfig3}
	\end{subfigure}
	\begin{subfigure}{.24\textwidth}
		\centering
		\includegraphics[width=\linewidth]{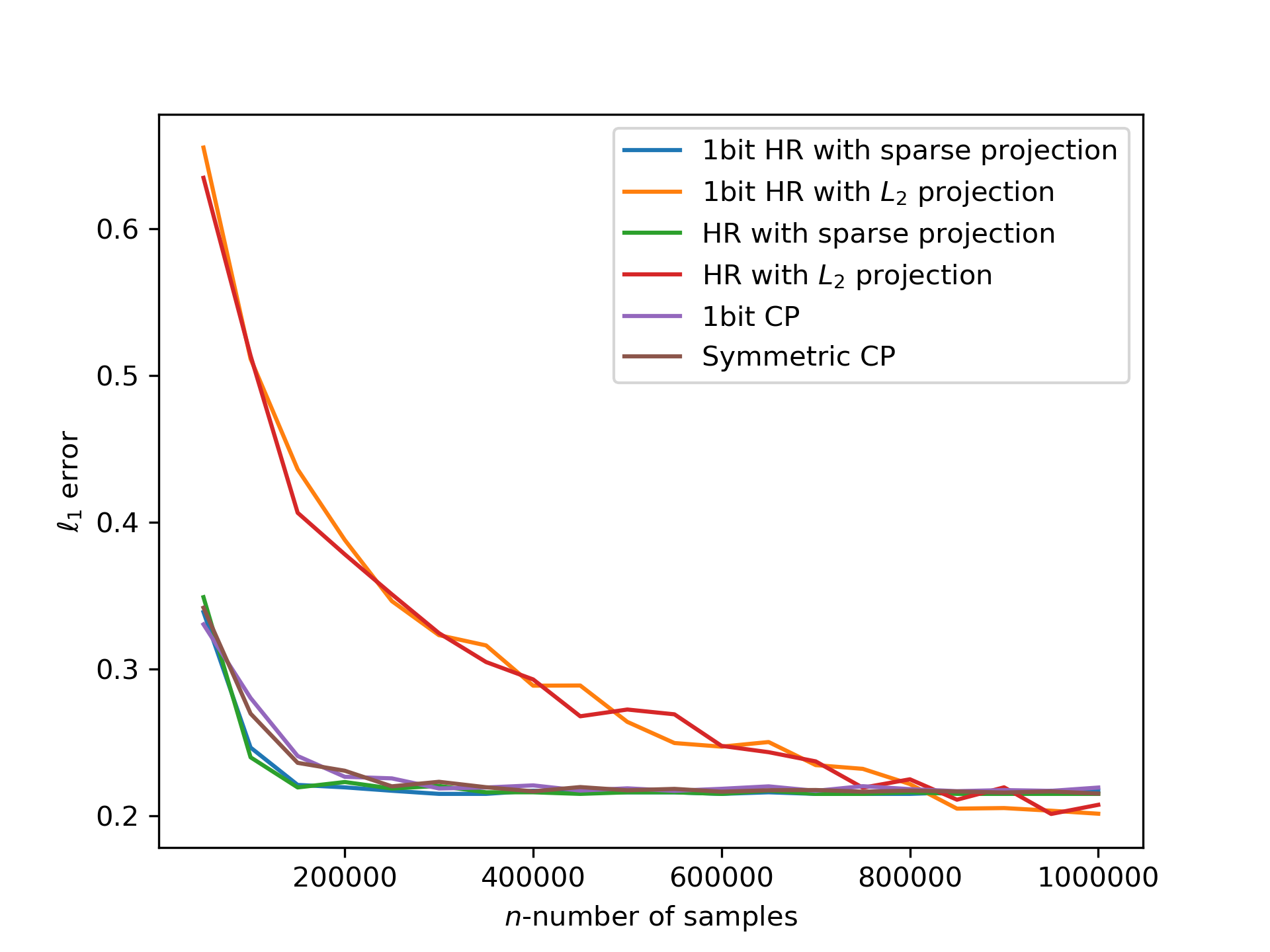}
		\caption{$\mathsf{Geo}(0.8)$}
		\label{fig:sfig4}
	\end{subfigure}

	\caption{$\ell_1$-error for $k = 10000, m=500, \varepsilon=1$. Sparse projection means $L_2$ projection onto simplex with sparsity constraint.}
	
	\label{fig:experimental_results1}
\end{figure*}
\begin{figure*}[!ht]
	\centering
	\begin{subfigure}{0.24\textwidth}
		\centering
		\includegraphics[width=\linewidth]{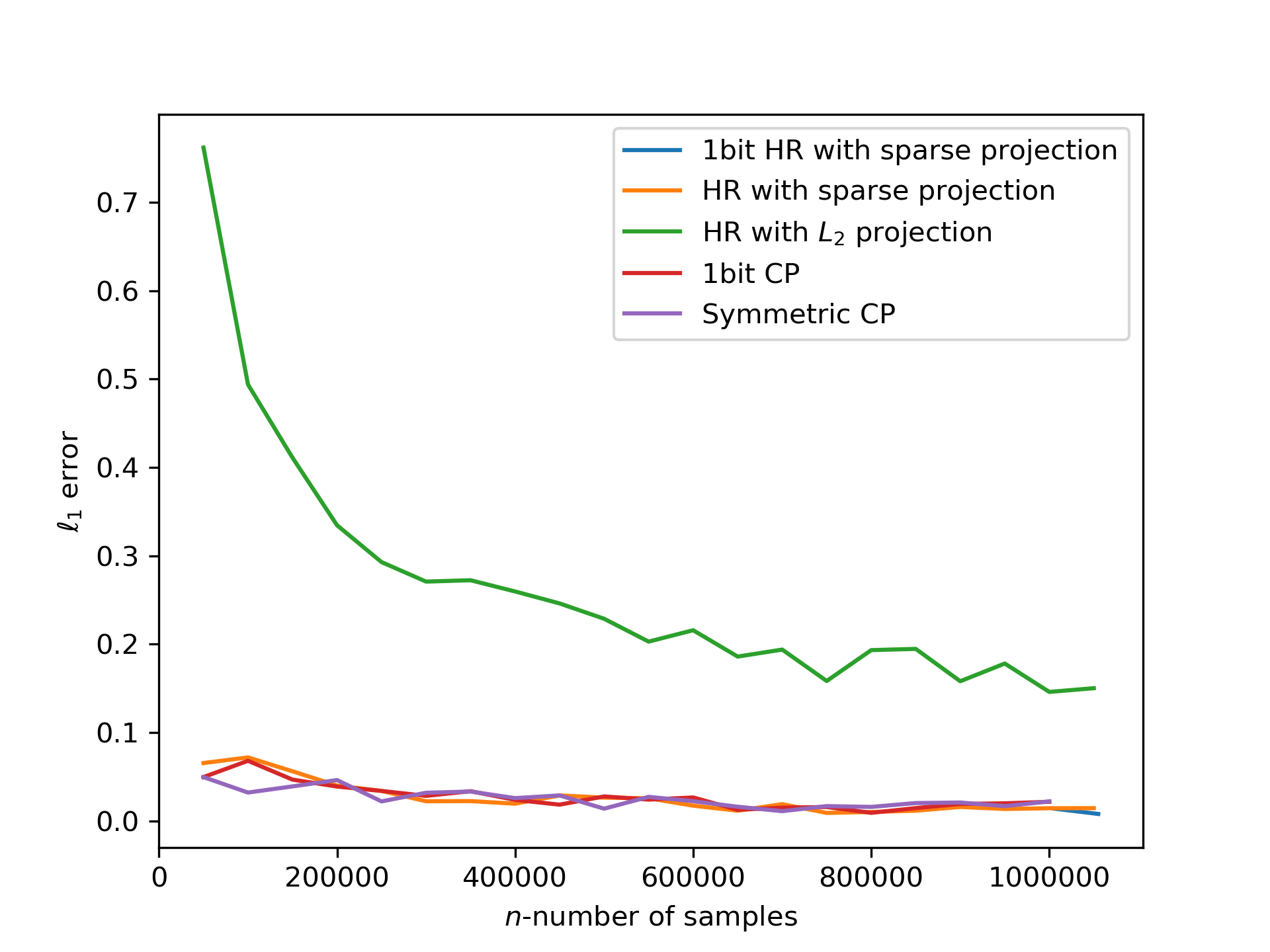}
		\caption{$\mathsf{Unif}(10)$}
		\label{fig:sfig9}
	\end{subfigure}
	\begin{subfigure}{.24\textwidth}
		\centering
		\includegraphics[width=\linewidth]{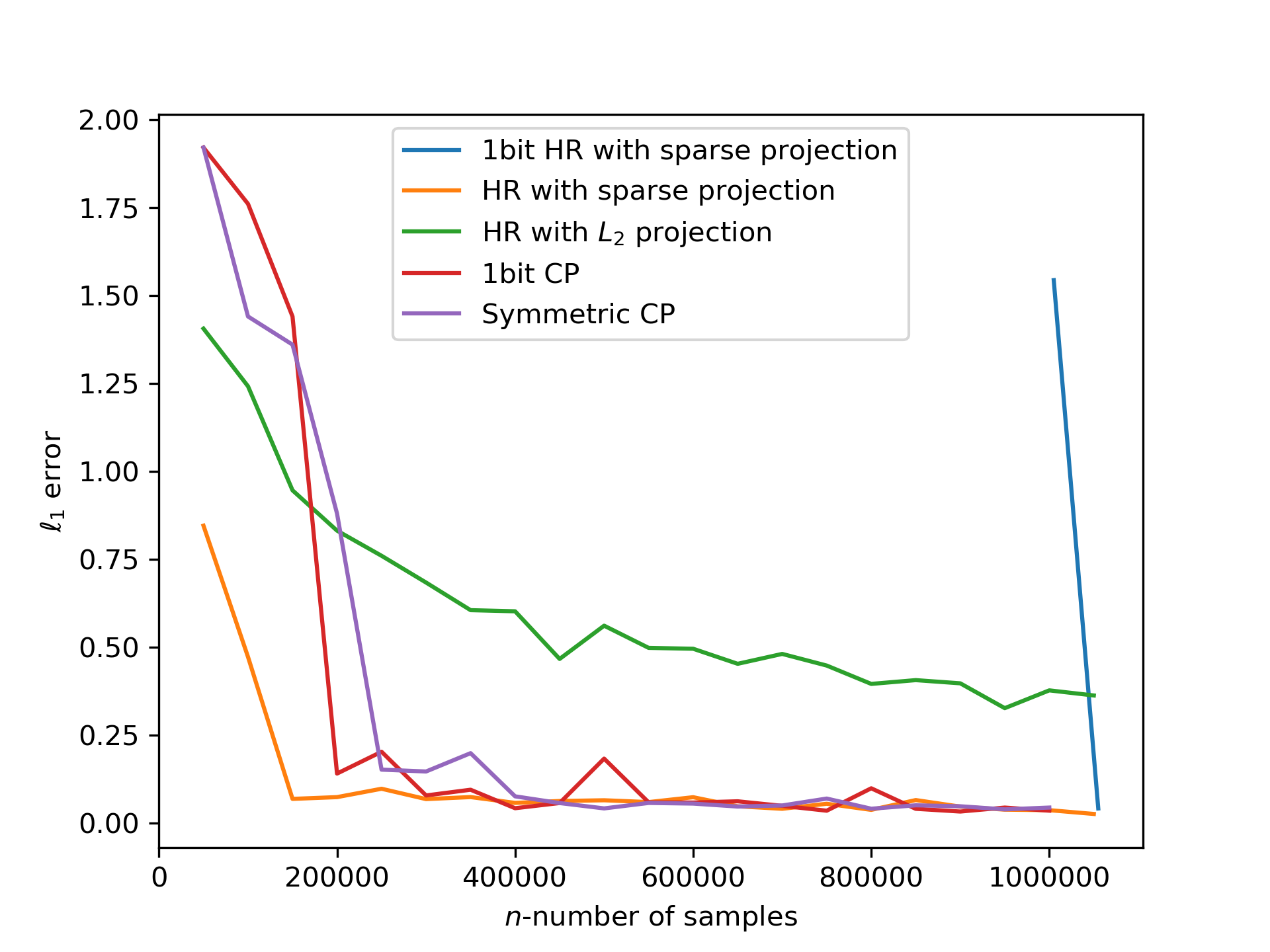}
		\caption{$\mathsf{Unif}(25)$}
		\label{fig:sfig10}
	\end{subfigure}
	\begin{subfigure}{.24\textwidth}
		\centering
		\includegraphics[width=\linewidth]{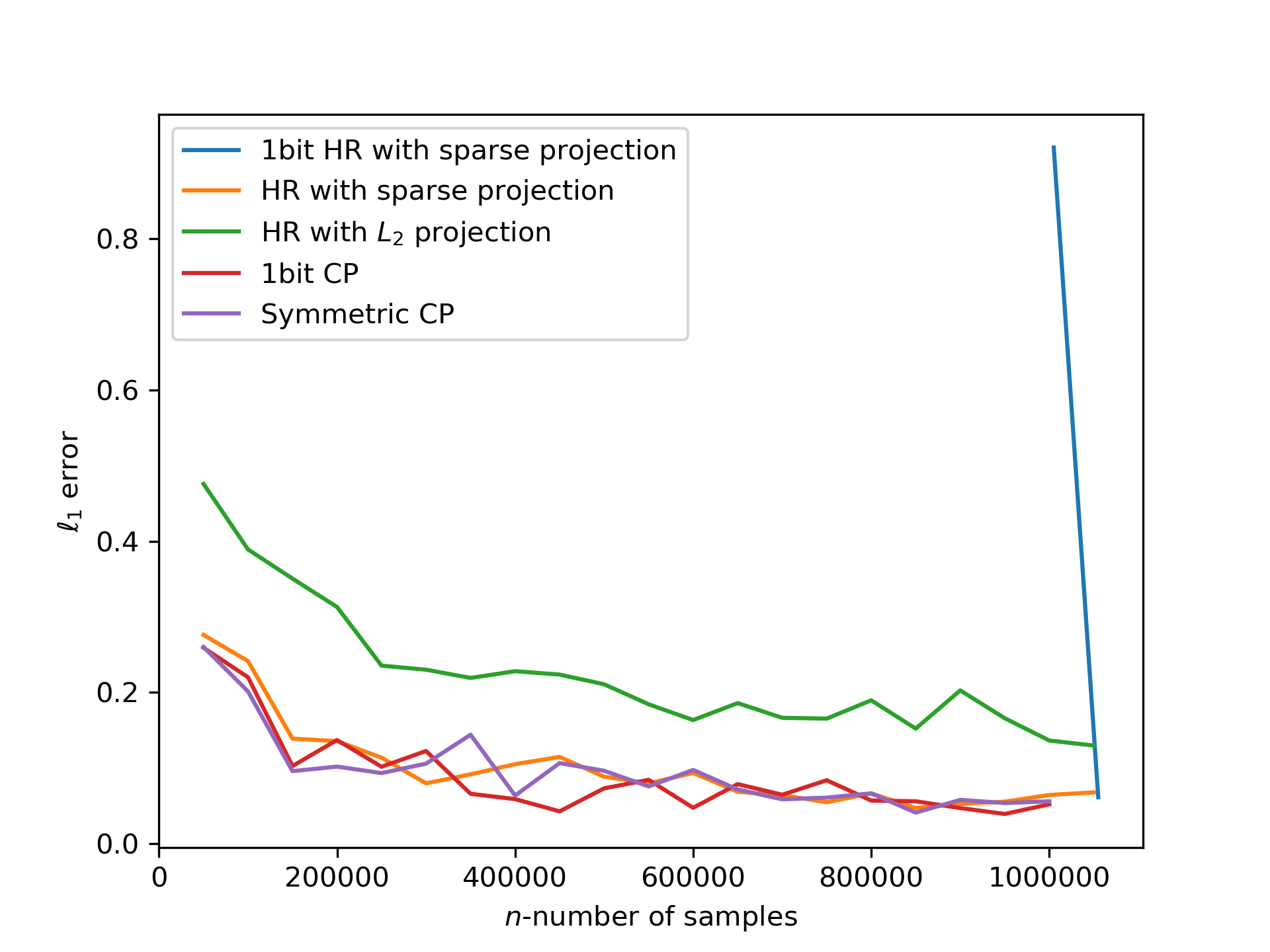}
		\caption{$\mathsf{Geo}(0.6)$}
		\label{fig:sfig11}
	\end{subfigure}
	\begin{subfigure}{.24\textwidth}
		\centering
		\includegraphics[width=\linewidth]{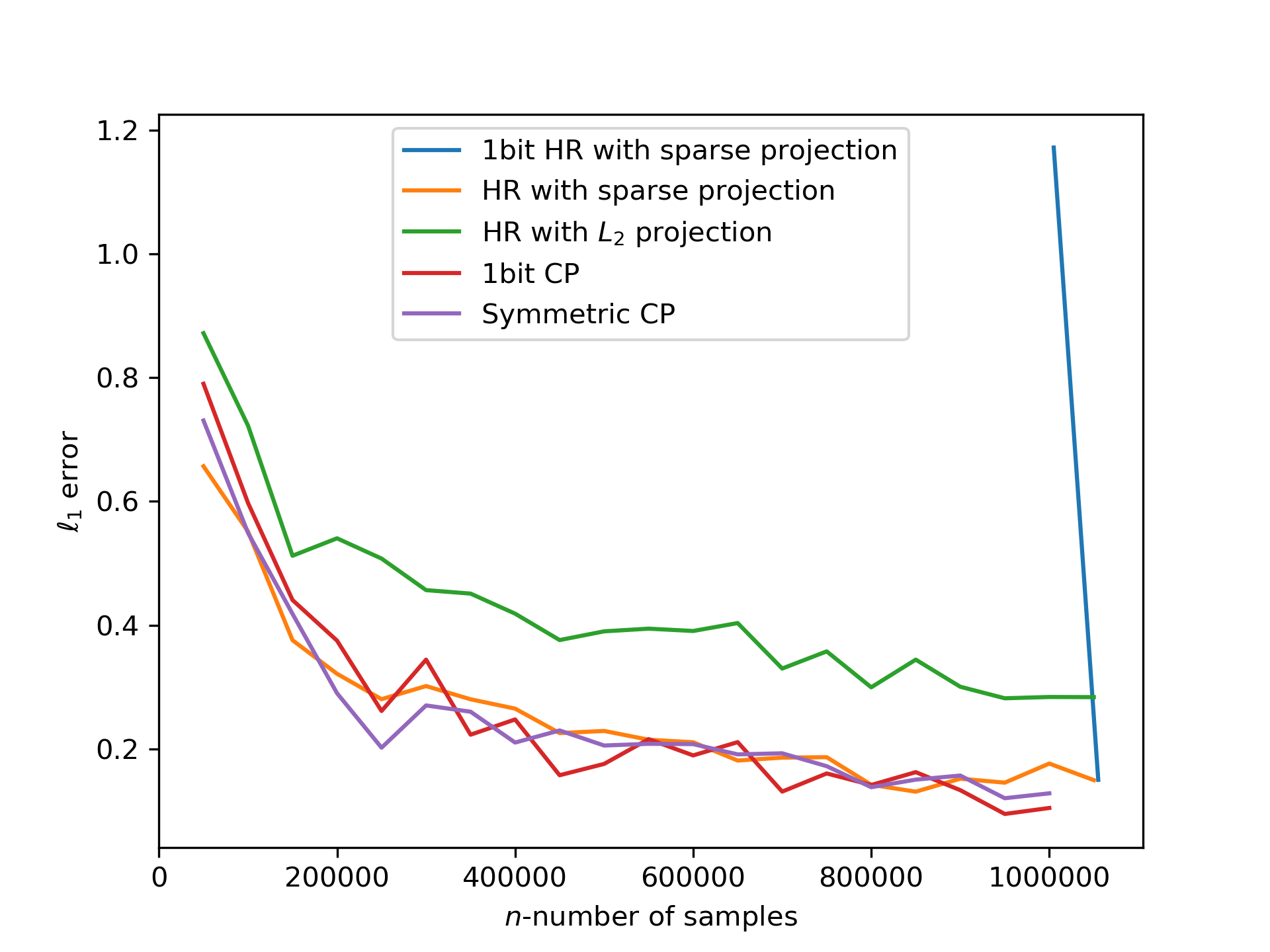}
		\caption{$\mathsf{Geo}(0.8)$}
		\label{fig:sfig12}
	\end{subfigure}
	
	\caption{$\ell_1$-error for $k = 1000000, m=500, \eps = 1$. For one bit HR, since the method requires the number of sample larger than $k$, thus the estimation result of one-bit HR starts from $n = 1000000$ in the figure.}
	\label{fig:experimental_results2}
	
\end{figure*}
The parameter $\beta$ is from Lemma~\ref{lemma:rcp_event_e}. Please refer to the appendix for more discussion. 
\section{Experiments}
We conduct experiments comparing our method with HR~\cite{acharya2018hadamard} and its one-bit version equipped with sparse projection. To implement our recovery process, we use the orthogonal matching pursuit (OMP) algorithm \cite{tropp2004greed}. We test the performances on two types of (approximately) sparse distributions: 1) geometric distributions $\mathsf{Geo}(\lambda)$ with $p(i) \propto (1-\lambda)^{i}\lambda$; and 2) sparse uniform distributions  $\mathsf{Unif}(s)$ where $|\mbox{supp}(p)|= s$ and $p(i) = \frac{1}{s}$ for $i \in \mbox{supp}(p)$. 

 In our experiments, the dimensionality of the unknown distribution is $k \in \{10000, 1000000\}$, and the value of $m$ in our method is set to $500$. The default value of the privacy parameter is $\eps=1$. 
 Here we provide the results of different algorithms on  $\mathsf{Geo}(0.8)$, $\mathsf{Geo}(0.6)$, $\mathsf{Unif}(10)$ and $\mathsf{Unif}(25)$. 
 
We record the estimation errors of different methods with varying number of samples $n \in \{50000, 100000, \cdots, 1000000\}$. Note that geometric distributions are not strictly sparse. For approximately sparse distributions, the sparsity parameter $s$ is chosen such that the distributions are roughly $(s, 0.1)$-sparse in our experiment. We assume that the value of $s$ is provided to the recovery algorithm. We simulate $10$ runs and report the average $\ell_1$ errors. The results are shown in Figure \ref{fig:experimental_results1} and Figure \ref{fig:experimental_results2}.

It can be seen from the numerical results that the performances of our compressive privatization approach are significantly better than the previous worst-case sample optimal methods like HR, which is aligned with our theoretical bounds. For small $k$, which is much smaller than sample size $n$, e.g. $k = 10000$, the one-bit HR with sparse projection is well-defined and the performance compared to our method is almost the same; this is not surprising as both method has the same (theoretical) sample complexity. Note, HR with sparse projection is better than with non-sparse projection. On the other hand, when $k$ is much larger e.g. $k = 1000000$, one-bit HR is not well-defined when the number of samples is less than $k$. In this case, we append zeros in the groups where there is no samples. However, the accuracy is much worse than our methods (see Figure \ref{fig:experimental_results2}). We note that HR with sparse projection still performs well in this case, but each user incurs $\log k$ bits of communication; our one-bit CP only needs one bit to achieve the same accuracy. For our symmetric CP method, the communication cost, which is $\log s+\log\log\frac{k}{s}$, is also lower than HR. 

\section{Conclusion}
In this paper, we study sparse distribution estimation in the local differential privacy model. We propose a compressive sensing based method, which overcome the limitations of the projection based method in \cite{acharya2021estimating}. For high privacy regime, we provide asymmetric and symmetric schemes, both of which achieves optimal sample and communication complexity. We also extend compressive privatization to medium privacy regime, and obtain near-optimal sample complexity for any privacy and communication constraints.
\begin{acks}
	This work is supported by National Natural Science Foundation of China Grant No. 61802069, Shanghai Science and Technology Commission Grant No. 17JC1420200, and Science and Technology Commission of Shanghai Municipality Project Grant No. 19511120700.
\end{acks}

\bibliographystyle{ACM-Reference-Format}
\bibliography{main}

\newpage
\appendix
\section{Missing proof from section \ref{sec:symmetric}}
\subsection{Proof of Lemma \ref{lemma:scp_privacy}}
\begin{proof}
	Observe that for any $x_1, x_2 \in [k]$, we have
	\begin{align*}
	    \max_{y\in[m]} \frac{Q(y|x_1)}{Q(y|x_2)} \leq \frac{n_{x_2}e^\varepsilon + m - n_{x_2}}{n_{x_1}e^\varepsilon + m - n_{x_1}} e^\varepsilon.
	\end{align*}
	By assumption, 
	\begin{align*}
	\frac{n_{x_2}e^\varepsilon + m - n_{x_2}}{n_{x_1}e^\varepsilon + m - n_{x_1}} &\leq \frac{((1 + \beta)\frac{m}{2})e^\varepsilon + m - ((1 + \beta)\frac{m}{2})}{((1 -  \beta)\frac{m}{2})e^\varepsilon + m - ((1 - \beta)\frac{m}{2})} \\&= \frac{1 + \beta\cdot \frac{e^\varepsilon -1}{e^\varepsilon + 1}}{1 - \beta\cdot \frac{e^\varepsilon -1}{e^\varepsilon + 1}} \leq \frac{1 + \beta/2}{1- \beta/2} \leq 1 + 2\beta.
	\end{align*}
	where the second inequality is from $\frac{e^\varepsilon - 1}{e^\varepsilon + 1} \leq \frac{1}{2}$ for $\varepsilon \in (0, 1)$ and the last inequality is from $0\leq \beta \leq 1$. It follows that 
	\begin{align*}
	\max_{x_1, x_2, y\in[k]} \frac{Q(y|x_1)}{Q(y|x_2)} &\leq \max_{x_1, x_2 \in [k]}\frac{n_{x_2}e^\varepsilon + m - n_{x_2}}{n_{x_1}e^\varepsilon + m - n_{x_1}} e^\varepsilon \\&\leq (1 + 2\beta)e^{\varepsilon} = e^{\varepsilon+\ln(1+2\beta)} \le e^{\eps+2\beta}.
	\end{align*}
	The last inequality is from $\ln (1+x) \le x$ for $x>-1$.
\end{proof}

\subsection{Proof of Lemma \ref{lemma:scp_noise_1}}
\begin{proof}
	By definition
 \begin{align*}
    \|e_1\|_2&=\norm{\frac{e^\varepsilon + 1}{\sqrt{m}(e^\varepsilon - 1)}J(D' - I) p}_{2} \\ 
&= \frac{ e^\varepsilon + 1}{\sqrt{m}(e^\varepsilon -1)}\sqrt{((D'- I)p)^T\cdot J^TJ\cdot( (D'-I)p)} \\
&\le \frac{ e^\varepsilon + 1}{e^\varepsilon -1}\sum_{i}|d_i' - 1| p_i,
 \end{align*}
 where $d_i' = \frac{m(e^\varepsilon + 1)/2}{n_i e^\varepsilon + m - n_i}$. By the assumption $n_i \in  (1 \pm \beta) \frac{m}{2}$ for all $i\in [k]$, we have 
 \begin{align*}
     |d_i'-1| &= \frac{|m/2-n_i|(e^{\eps}-1)}{n_i e^\varepsilon + m - n_i} \leq \frac{\beta m (e^{\eps} - 1)/2}{n_i e^\varepsilon + m - n_i}  \\&\leq \frac{\beta m (e^{\eps} - 1)/2}{(1-\beta)\frac{m}{2} e^\varepsilon + \frac{m}{2} -\frac{\beta m}{2}} = \frac{\beta(e^{\eps} -1)}{(1-\beta)(e^{\eps} +1)}.
 \end{align*}
 Thus, $\|e_1\|_2 \le \frac{\beta}{1-\beta}$, which completes the proof.
\end{proof}
\subsection{Proof of Lemma \ref{lemma:scp_noise_2}}
\begin{proof}
Let $Y_1, Y_2, \cdots, Y_n$ be the privatized samples received by the server.  We have, for each $i\in[m]$, $ \hat{q}_i = \sum_{j = 1}^{n} \frac{\mathbb{I}(Y_j = i)}{n}$,
where $\mathbb{I}$ is the indicator function.Thus $\mathbb{E}[\hat{q}_i] = q_i$ and $\operatorname{Var}[q_i] = \frac{q_i (1-q_i)}{n} \le \frac{q_i }{n} $. It follows that
\begin{align*}
   \mathbb{E}\left[\norm{\hat{q} - q}_2\right] &\leq \sqrt{\mathbb{E}\left[\norm{\hat{q}-q}^2\right]} = \sqrt{\sum_i \mbox{Var}(\hat{q}_i)}   \leq \sqrt{\frac{1}{n}}
\end{align*}
where the first inequality is from Jensen's inequality. Multiplying $\frac{\sqrt{m}e^\varepsilon + 1}{e^\varepsilon -1}$ on both sides of the inequality will conclude the proof.
\end{proof}
\subsection{Proof of Theorem \ref{theorem:scp_reconstruction_error}}
\begin{proof} 
Let $p' = D'p$. By definition of $\hat{p}'$, we have
\begin{align*}
    \|p - \hat{p}'\|_2 = \|D'^{-1}p - D'^{-1}f\|_2 \leq \max_{i}\frac{1}{d_i'} \cdot \norm{p' - f}_2
\end{align*}
where $d_i' = \frac{m(e^\varepsilon +1)/2}{n_i e^\varepsilon + m - n_i}$ and $f$ is the output of the recovery algorithm (step 7 in Algorithm \ref{algorithm:scp}). Since $\forall i \in [k], n_i \leq (1+\beta)\cdot \frac{m}{2}$ for $0\leq \beta\leq 1$, then we have 
\begin{align*}
    \norm{p - \hat{p}'}_2 \leq \max_{i}\frac{1}{d_i'} \cdot \norm{p' - f}_2 &\leq (1 + \beta \cdot \frac{e^\varepsilon - 1}{e^\varepsilon + 1}) \norm{p' - f}_2
    \\&\leq (1 + \frac{\beta}{2}) \norm{p' - f}_2. 
\end{align*}
By Lemma \ref{lemma:cs_recovery},  we have 
\begin{align*}
    \norm{p' - f}_2 &\leq \frac{C}{\sqrt{s}}\|p'-[p']_s\|_1 +  D\norm{e_1 + e_2}_2 \\&\leq \max_{i}\frac{1}{d_i'}\cdot \frac{C}{\sqrt{s}}\|p-[p]_s\|_1 + D\left(\norm{e_1}_2 + \norm{e_2}_2\right).
\end{align*} 
Note $\max_{i}\frac{1}{d_i'} \le (1 + \beta/2)$ as $n_i \leq (1+\beta)\cdot \frac{m}{2}$ for all $i$. By Lemma \ref{lemma:scp_noise_1}, \ref{lemma:scp_noise_2}, we get
\begin{align*}
    \E[\norm{p - \hat{p}'}_2] &\leq \left(1 + \frac{\beta}{2}\right)\left(2D\beta+ \frac{D(e^\varepsilon + 1)}{e^\varepsilon -1}\sqrt{\frac{m}{n}} \right) \\&+\left(1+\frac{\beta}{2}\right)^2\left( \frac{C}{\sqrt{s}}\|p-[p]_s\|_1\right),   
\end{align*}
which proves the theorem.
\end{proof}

\section{Missing proof from section \ref{sec:median}}

\subsection{Proof of Lemma \ref{lemma:rcp_event_e}}
\begin{proof}
By symmetry, we only consider the sparsity of one specific block, say the first one. Let $k_1 = \frac{k}{L}$ be the size of a block. Let $s_1$ denote the sparsity, i.e.\ the number of non-zero entries, of the first block. Then, we have
\begin{align*}
    s_1 = \sum_{j = 1}^{k_1} \mathbf{1}\{p'_j \neq 0\}
\end{align*}
where $\mathbf{1}\{p'_j \neq 0\}$ is an indicator to describe whether the $j$-th position of $p'$ is nonzero. By direct calculation, we can get that $\E[\mathbf{1}\{p'_j \neq 0\}] ={\tbinom{k -1}{s - 1}}/{\tbinom{k}{s}} = \frac{s}{k}$. Thus, $\E[s_1] = \frac{sk_1}{k} =\frac{s}{L}$. Since $p'$ is a random permutation of $p$, $\mathbf{1}\{p'_1 \neq 0\}, \cdots, \mathbf{1}\{p'_{k_1} \neq 0\}$ are negatively associated (NA) \cite{wajc2017negative}. By Chernoff-Hoeffding bounds for NA variables \cite{wajc2017negative, dubhashi1996balls}, we can get that 
\begin{align}
\label{eq:chernoff_na}
    \operatorname{Pr}\left[s_1 \geq (1 + \beta)\E[s_1] \right] &\leq \left(\frac{e^\beta}{(1 + \beta)^{(1 + \beta)}}\right)^{\E[s_1]} \notag\\&= e^{\E[s_1]\left(\beta - (1 + \beta)\ln{(1 + \beta)}\right)}
\end{align}
It can be easily verified that 
\begin{align}
\label{eq:beta_ineq}
    \beta - (1+\beta)\ln{(1 + \beta)} \leq \left\{\begin{array}{cc}
      -\frac{\beta^2}{4}   &  \beta <= 4, \\
       -\frac{\beta}{4}  & \mbox{otherwise} 
    \end{array}\right. \leq -\frac{1}{4} \min\{\beta^2, \beta\}
\end{align}

Combining \eqref{eq:chernoff_na}, \eqref{eq:beta_ineq} and $\E[s_1] = \frac{s}{L}$ yields that
\begin{align}
     \operatorname{Pr}\left[s_1 \geq (1 + \beta)\frac{s}{L} \right] \leq e^{-\frac{\min\{\beta^2, \beta\}s}{4L}}
\end{align}
The proof is then completed by applying union bound over all $L$ sections. 
\end{proof}

\subsection{Proof of Lemma \ref{lemma:rcp_q_j_variance}}
\begin{proof}
For any $j\in [m]$, we have
\begin{align}
\label{eq:total_expectation}
    \mathbb{E}[\hat{Q}_j(X_i')] = \mathbb{E}_p[\mathbb{E}_{\eps}[\hat{Q}_j(x)|X_i' = x]]
\end{align}
When $X_i'$ is fixed to be $x \in \mathcal{Q}$ where $\mathcal{Q} := \{\pm e_1, \cdots, \pm e_L\}$ and we only consider the randomness from the privatization, we have
\begin{align}
\label{eq:privatization_randomness}
    \mathbb{E}_{\eps}\left[\hat{Q}_j(x)\right] &= \frac{e^\eps}{e^\eps + 2L - 1} Q_j(x)  + \sum_{Q' \in \mathcal{Q}\backslash\{Q_j(X_i')\}} \frac{1}{e^\eps + 2L - 1}\notag \\
    &= \frac{(e^\eps - 1) }{e^\eps + 2L - 1}\cdot Q_j(x)
\end{align}
By the definition of $Q_j$, we can get 
\begin{align}
\label{eq:p_randomness}
    \mathbb{E}\left[Q_{j}\left(X_{i}'\right)\right]=\left[\begin{array}{c}
\left(A_2\right)_{j} \cdot {p'}^{(1)} \\
\left(A_2\right)_{j} \cdot {p'}^{(2)} \\
\vdots \\
\left(A_2\right)_{j} \cdot {p'}^{\left(L\right)}
\end{array}\right]
\end{align}
Combining \eqref{eq:total_expectation}, \eqref{eq:privatization_randomness} and \eqref{eq:p_randomness} yields that
\begin{align}
\label{eq:Q_j_hat}
    \mathbb{E}[\hat{Q}_{j}(X_i')] = \frac{(e^\eps - 1) }{e^\eps + 2L - 1}  \left[\begin{array}{c}
\left(A_2\right)_{j} \cdot {p'}^{(1)} \\
\left(A_2\right)_{j} \cdot {p'}^{(2)} \\
\vdots \\
\left(A_2\right)_{j} \cdot {p'}^{\left(L\right)}
\end{array}\right]
\end{align}
Recall that $q = A\cdot p'$ and $A = A_1 \otimes A_2$. For $j' \equiv j$ (mod $m$) and $j' = j + (t - 1)m $, by \eqref{eq:Q_j_hat}, we have
\begin{align*}
    \mathbb{E}[\hat{q}_{j'}] &= \frac{m}{n} \cdot \frac{e^\eps + 2L - 1}{e^\eps - 1}\sum_{i \in S_j} (A_1)_t \cdot  \mathbb{E}[\hat{Q}_j(X_i')] \\
    &= (A_1)_t \cdot \left[(A_2)_{j}\cdot p^{'(1)}, \cdots, (A_2)_{j}\cdot p^{'(L)}\right] = q_{j'}
\end{align*}
where the last equality is from the definition of kronecker product. Hence, $\hat{q}_{j'}$ is an unbiased estimator for $q_{j'}$. Thus, 
\begin{align*}
    \mathbb{E}[(q_{j'} - \hat{q}_{j'})^2] = \var(q_{j'}) &= \frac{m}{n}\left(\frac{e^\eps + 2L - 1}{e^\eps - 1}\right)^2 \var\left((A_1)_t \cdot  \hat{Q}_j(X_i')\right) \\
    &\leq \frac{m}{n} \left(\frac{e^\eps + 2L - 1}{e^\eps - 1}\right)^2
\end{align*}
where the inequality is from that $(A_1)_t \cdot \hat{Q}_j(X_i')$ only takes value in 
$\{+1, -1\}$. The proof is completed.
\end{proof}

\subsection{Proof of Theorem \ref{thm: rcp_medium_result}}

\begin{proof}
From the analysis of estimation error in section \ref{sec:median}, we know that
\begin{align*}
\E[\|p - \hat{p}\|_2] \leq \E[\|p' - \hat{p}'\|_2] \leq \frac{1}{\sqrt{mL}} \E[\|\hat{q} - q\|_2] 
\end{align*}
By Lemma \ref{lemma:rcp_q_j_variance}, we can get
\begin{align*}
    \E[\|p - \hat{p}\|_2] &\leq \E[\norm{p' - \hat{p}'}_2] \leq \frac{1}{\sqrt{mL}} \E[\|\hat{q} - q\|_2] \\&\overset{(a)}{\leq} \sqrt{\frac{1}{mL}\E\left[\norm{\hat{q} - q}_2^2\right]} \\
    &= \sqrt{\frac{1}{mL}\sum_{j' \in [mL]}\E\left[(\hat{q}_{j'} - q_{j'})^2\right]} \\
    &\overset{(b)}{\leq}\sqrt{\frac{m}{n}}\left(\frac{e^\eps + 2L - 1}{e^\eps - 1}\right) \overset{(c)}{\leq} \sqrt{\frac{m}{n}}\left(\frac{3e^\eps - 1}{e^\eps - 1}\right)
\end{align*}
where $(a)$ is from Jensen's inequality and $(b)$ is from Lemma \ref{lemma:rcp_q_j_variance} and $(c)$ is from $L = \min\{e^\eps, 2^b\}$.
\begin{enumerate}
    \item $\eps = O(1)$. In this case, we can set $L = 1$ directly and the communication is 
    $1$ bit now. The event $\mathcal{E}$ then holds with probability $1$, hence $m = s\log(k / s)$. Since $\eps = O(1)$, $\frac{3e^\eps - 1}{e^\eps - 1} = O(\frac{1}{\eps})$. Thus $\E[\norm{p - \hat{p}}_2] = O\left(\sqrt{\frac{s\log(k / s)}{n\eps^2}}\right)$, which is the same error bound as that in one-bit CP for high privacy. Note that $A_1$ is $1$ now, $A = A_1 \otimes A_2 = A_2 \in \mathbb{R}^{m \times k}$ is a Rademacher matrix. Therefore, for high privacy, if we set $L = 1$, our scheme is exactly one-bit CP.
    \item ${\eps} = \omega(1)$. We mainly  consider medium privacy case, where $\eps = \omega(1)$ and  $e^{\eps} \leq s$. In this case, $\frac{3e^\eps - 1}{e^\eps - 1} = O(1)$. Hence, we have that $\E[\norm{p - \hat{p}}_2] = O\left(\sqrt{\frac{(1 + \beta)s\log(k / (1+\beta)s)}{nL}}\right)$ (note $m=O\left(\frac{(1 + \beta)s\log(k / (1+\beta)s)}{L} \right)$). When $n \geq c \cdot \frac{(1 + \beta)s\log(k / (1+\beta)s)}{L\alpha^2}$ for some large enough constant $c$, the expected $\ell_2$ error is at most $\alpha$. For $\ell_1$ error, we have
    \begin{align*}
        \E[\norm{p - \hat{p}}_1] \leq 2\sqrt{2s}\E[\norm{p' - \hat{p}'}_2]
    \end{align*}
    Thus to achieve an $\ell_1$ error of $\alpha$, it's sufficient to get an estimate with $\alpha' = \alpha / 2\sqrt{2s}$ for $\E[\norm{p' - \hat{p}'}_2]$. The sample complexity is $O\left(\frac{(1 + \beta)s^2\log(k / (1+\beta)s)}{L\alpha^2}\right)$. Since the event $\mathcal{E}$ holds with probability at least $1 - Le^{-\frac{\min\{\beta^2, \beta\}s}{4L}}$ and the RIP condition holds with probability at least $1 - e^{-m}$, by union bound, we can achieve the sample complexity above with probability $1 - Le^{-\frac{\min\{\beta^2, \beta\}s}{4L}} - e^{-m}$, where $m = \frac{(1 + \beta)s\log(k / (1 + \beta)s)}{L}$. When $\min\{\beta^2, \beta\} = \frac{4L\log(L /\delta)}{s}$, the error probability from $\mathcal{E}$ is less than $\delta$. If $L\log(L/ \delta) = O(s)$, then $\min\{\beta^2, \beta\} = O(1)$ which means $\beta = O(1)$. In this case, with probability $1 - \delta - e^{-m}$, the sample complexity for $\ell_2$ error is $O\left(\frac{s\log(k/s)}{L\alpha^2}\right)$ and for $\ell_1$ error is $O\left(\frac{s^2\log(k/s)}{L \alpha^2}\right)$. When $s \ll L\log(L/\delta)$ and $L \leq s$, $\min\{\beta^2, \beta\} = O(\log(s /\delta))$ which means $\beta = O(\log(s/\delta))$. For general distribution under medium privacy regime, the sample complexity for $\ell_1$ error is at least $\Omega(\frac{k^2}{L\alpha^2})$ \cite{chen2020breaking}, which implies a lower bound of $\Omega(\frac{s^2}{L\alpha^2})$ for $s$-sparse distributions. Thus the sample complexity blows up by at most a logarithmic factor. 
\end{enumerate}
\end{proof}

\end{document}